\documentclass{comsoc2020}

    \long\def\symbolfootnote[#1]#2{\begingroup%
    \def\thefootnote{\fnsymbol{footnote}}\footnote[#1]{#2}\endgroup} 

\usepackage[dvipsnames]{xcolor}
\usepackage{times}
\usepackage{soul}
\usepackage{url}
\usepackage[hidelinks]{hyperref}
\usepackage[utf8]{inputenc}
\usepackage[small]{caption}
\usepackage{graphicx}
\usepackage{amsmath}
\usepackage{amsthm}
\usepackage{booktabs}
\usepackage{tikz}
\usetikzlibrary{backgrounds,positioning,fit,decorations.pathreplacing}

\usepackage[noend,ruled,linesnumbered]{algorithm2e}
\usepackage[round]{natbib}
\usepackage{mathtools}
\usepackage{amssymb}
\usepackage{cleveref}
\usepackage{graphicx} 
\usepackage{subcaption}
\usepackage{wasysym} 

\newtheorem{theorem}{Theorem}
\newtheorem{corollary}[theorem]{Corollary}
\newtheorem{lemma}[theorem]{Lemma}
\newtheorem{proposition}[theorem]{Proposition}
\newtheorem{example}{Example}
\newtheorem{definition}[theorem]{Definition}

\newcommand{\avg}{\text{avg}} 
\newcommand{\gr}{\kappa} 
\newcommand{\R}{\mathcal{R}} 
\newcommand{\Set}[1]{\{ #1 \}} 
\newcommand{\mc}{\mathrm{mc}} 
\newcommand{\tsc}{\mathrm{sc}} 
\newcommand{\credit}{\mathrm{\cent}} 

\title{Dynamic Proportional Rankings\symbolfootnote[1]{A short version of this paper appears in the proceedings of IJCAI-2021 (Israel and Brill, 2021).}}
\author{Jonas Israel and Markus Brill\\ 
Efficient Algorithms Research Group\\
Technische Universität Berlin}

\begin{document}

\nocite{IsBr21b}

\begin{abstract}
    Proportional ranking rules aggregate approval-style preferences of agents into a collective ranking 
    such that groups of agents with similar preferences are adequately represented. 
    Motivated by the application of live Q\&A platforms, where submitted questions need to be ranked based on the interests of the audience, we study a dynamic extension of the proportional rankings setting. In our setting, the goal is to maintain the proportionality of a ranking when alternatives (i.e., questions)---not necessarily from the top of the ranking---get selected sequentially. 
    We propose generalizations of well-known aggregation rules to this setting and study their monotonicity and proportionality properties. 
    We also evaluate the performance of these rules experimentally, using realistic probabilistic assumptions on the selection procedure.  
\end{abstract}

\section{Introduction}
\label{sec:intro}

From “ask-me-anything” sessions to panel discussions and town hall meetings, an increasing number of both virtual and in-person discussion formats are enhanced by digital tools that aim to make the event more interactive and responsive to the audience. Using live Q\&A platforms such as 
\textit{slido} (\url{https://www.sli.do}), 
\textit{Mentimeter} (\url{https://www.mentimeter.com}) 
or \textit{Pigeonhole Live} (\url{https://pigeonholelive.com}), participants in the audience can submit questions and upvote questions submitted by others; a moderator then selects the most popular questions for the discussion. By reducing barriers to participation (e.g., by allowing anonymous submissions), these tools aim to better represent the diversity in the audience.

The moderator of the discussion is presented with an aggregated list, in which audience questions are ranked by popularity (i.e., number of upvotes).
Based on this ranking, the moderator then picks the next question.
When selecting a question, it is usually not required to follow the ranking strictly; rather, the choice is at the moderator's discretion, allowing him or her to take into account other factors such as discussion flow, etc. That being said, it is generally expected that questions at the top of the ranking are more likely to be selected than questions further down in the list. 
After a question has been selected, it is removed from the ranking.

Ranking questions solely by popularity, though intuitively appealing, has a major downside: minority opinions might go completely unrepresented, even when the minority makes up a substantial proportion of the audience. To illustrate this phenomenon, which is often referred to as ``tyranny of the majority,'' consider a situation in which the audience is composed of two groups. One group makes up 60\% of the entire audience and is only interested in questions related to topic~$A$; the remaining  40\% of participants are only interested in questions on a different topic $B$. Now, assuming that sufficiently many questions on topic $A$ have been submitted, and that participants only upvote questions related to their own interest, questions on topic $B$ are unlikely to appear anywhere near the top of the ranking, which is populated exclusively by questions on topic $A$. As a consequence, questions on topic $B$ are very unlikely to be selected, despite the fact that these questions are supported by 40\% of the audience. 
 
In this paper, we propose an approach to avoid the  problem of underrepresenting minority opinions. Specifically, we model the scenario described above as a \textit{proportional representation} problem and employ ranking algorithms based on (approval-based) proportional voting rules \citep{ABC+16a,BFJL16a}. 
The algorithms we consider aggregate the upvotes of the participants into a \textit{proportional} ranking over questions, such that each minority (i.e., group of participants with similar preferences) is represented in the ranking to an extent that is proportional to the group's size.
Whenever a question is selected by the moderator, our methods dynamically recompute the ranking, pushing questions supported by underrepresented groups closer to the top.

At a technical level, our point of departure is the theory of proportional rankings \citep{SLB+17a}, which studies how a collective ranking over a set of alternatives can be constructed in such a way that majority and minority opinions are represented adequately.  
The question we are interested in is how proportional ranking algorithms can be adapted to the dynamic setting. More specifically, we ask:
\begin{quote}
\textit{How can the proportional representativeness of a collective ranking be maintained in a dynamic setting, where alternatives get selected sequentially?} 
\end{quote}
To answer this question, we consider two well-known aggregation rules dating back to the late 19th century: sequential \citet{Phra94a} and sequential PAV \citep{Thie95a}. These two rules, together with a few variants of the latter, performed best in the analysis conducted by \citet{SLB+17a}. For both rules, we propose two distinct generalizations to the sequential selection setting: a \textit{dynamic} variant and a \textit{myopic} variant (see \Cref{sec:dynamic_rules} for details). 
As a benchmark, we also consider the rule that simply orders questions by the number of received upvotes.

\paragraph{Our Contribution.}
In this paper, 
(i) we formalize the setting of dynamic ranking rules and generalize the rules of Phragmén and Thiele to this setting (\Cref{sec:dynamic_rules});
(ii)~we define a notion of satisfaction monotonicity and analyze to what extent the considered rules satisfy it (\Cref{sec:impl_mono});
(iii) we provide theoretical bounds regarding two different proportionality notions (\Cref{sec:proportionality});
and (iv) we experimentally evaluate our dynamic ranking rules (\Cref{sec:experiments}).
Omitted proofs and further details can be found in the appendix. 

\paragraph{Related Work.}
Proportional representation is a fundamental desideratum in multiwinner elections \citep{Monr95a,FSST17a,LaSk20a}. For approval preferences in particular, a wide variety of proportionality axioms have been studied \citep{ABC+16a,SFF+17a,Jans18a,PeSk20a}. Proportionality in the context of rankings has been considered in the aforementioned paper by \citet{SLB+17a} and (for linear preferences) by \citet{Schu11b}.

Notions of fairness over multiple elections among a fixed set of voters have received considerable attention in previous years. This line of work includes, e.g., the study of long-term fairness over different decisions \citep{freeman2017fair,lackner2020perpetual}, single decisions under changing preferences \citep{tennenholtz2004transitive,boutilier2012dynamic,parkes2013dynamic,oren2014online,hemaspaandra2017complexity}, and storable votes \citep{casella2005storable,casella2012storable}. 

In a practical attempt to avoid the underrepresentation of minorities, the  live Q\&A app \textit{SpeakUp} (\url{https://speakup.digital/}) allows audience members to add \textit{attributes} (relating to, e.g., gender or education) to submitted questions. The moderator can then manually filter questions with attributes that have been underrepresented in the discussion.  
Requiring organizers to identify relevant attributes poses the risk of overlooking important subgroups or introducing unwanted biases; it also presumes the willingness of participants to reveal potentially sensitive information.
In contrast, the ranking algorithms considered in this paper do not require attributes in order to ensure the representation of minority opinions.

\section{Preliminaries}
\label{sec:prelims}

We briefly introduce some basic concepts from the theory of approval-based preference aggregation; for details, see the survey by \citet{LaSk20a}. Let $C$ be a finite set of candidates and $N=\{1, \ldots, n\}$ a finite set of voters. An \textit{(approval) profile} $A = (A_1, \ldots, A_n)$ is a list that contains, for each $i \in N$, the approval set $A_i \subseteq C$ of voter $i$.
Given an approval profile $A$ and a candidate $c \in C$, we let $N_c = \{i \in N : c \in A_{i}\}$ denote the \emph{supporters} of $c$. The \emph{approval score} of $c$ is given by $|N_c|$.
In the motivating application, $C$ consists of all submitted questions and $A_i$ contains the questions that have been upvoted by participant~$i$.

To measure satisfaction of a group of voters $V \subseteq N$ with a set $S \subseteq C$ of candidates, we often use the \emph{average satisfaction} of $V$ with $S$, i.e., 
\[ \avg_V(S) = \frac{1}{|V|} \cdot \sum_{i \in V} |A_i \cap S|. \]

For a finite set $S$, we let $\mathcal{L}(S)$ denote the set of all linear orders, or \emph{rankings}, over $S$. We often write a ranking \mbox{$r \in \mathcal{L}(S)$} as a sequence $r = (r_1, r_2, \ldots, r_{|S|})$, and for $j \le |S|$, we let $r_{\leq j}$ denote the set $\Set{r_1, r_2, ... , r_j}$ of the first $j$ elements in $r$.

An approval-based \textit{ranking rule} maps an approval profile $A$ to a ranking $r \in \mathcal{L}(C)$ of all candidates.
We will make use of the following three (non-dynamic) ranking rules.\footnote{All rules may encounter ties; we assume that a priority ordering over candidates is used as a tiebreaker. In the motivating example, the submission time of a question yields a natural priority ordering.}

\paragraph{Approval Voting (AV).}
AV ranks the candidates according to their approval score. This rule is not proportional and we use it mainly as a benchmark. 

\paragraph{Sequential PAV (seqPAV).}
This rule ranks candidates iteratively, in each iteration choosing an unranked candidate maximizing the marginal contribution in terms of weighted voter satisfaction. 
Formally, for a subset $S \subseteq C$ of candidates, define
\[\tsc(S) = \sum_{i \in N} \sum_{j=1}^{|A_i \cap S|} \frac{1}{j}.\]
If $k$ candidates have already been ranked, the marginal contribution of an unranked candidate $c$
is given by
$\mc(c) = \tsc(r_{\leq k} \cup \{c\}) - \tsc(r_{\leq k})$.

\paragraph{Sequential Phragmén.}
This rule can be described in terms of voters buying candidates.\footnote{An equivalent formulation of this method is in terms of a load balancing procedure \citep{Jans16a,BFJL16a}.}  
Every candidate costs 1 credit. All voters start without any credits but earn them continuously over time (at a constant and identical rate).
As soon as a group of voters who all approve the same candidate~$c$ together own~1 credit, they immediately buy that candidate; at this point, their balance is reset to~$0$ and candidate $c$ is added in the next position of the ranking. This is done until all candidates are ranked.

\section{Dynamic Ranking Rules}
\label{sec:dynamic_rules}

In this section, we formally introduce the setting of dynamic ranking rules and we adapt existing (non-dynamic) ranking rules to this setting. The input of a dynamic ranking rule consists of two parts: an approval profile and a (potentially empty) sequence of candidates that have already been ``implemented'' or ``executed'';  the output is a ranking of all not-yet-implemented candidates. To formalize this notion, we let $X = (x_1, x_2, \ldots, x_j)$ denote the sequence of implemented candidates (where $j \in \{0, \ldots, |C|\}$); whenever the order of elements in $X$ does not matter, we slightly abuse notation and treat $X$ as the set $X = \{x_1, x_2, \ldots, x_j\}$. 

\begin{definition}
An \emph{(approval-based) dynamic ranking rule}~$\R$ maps a profile $A$ and a sequence $X = (x_1, x_2, \ldots, x_j)$ of candidates to a ranking $\R(A, X) \in \mathcal{L}(C\setminus X)$. 
\end{definition}
 
Applying a dynamic ranking rule to a sequential selection process (as outlined in the introduction) is now straightforward: At the beginning, when no candidate has yet been implemented, $X=()$ and the ranking $\R(A,())$ ranks all candidates in $C$. Given this ranking, a \textit{decision maker (DM)} selects an alternative $x_1 \in C$ to be implemented. The updated ranking of the remaining candidates is then given by $\R(A,(x_1))$, and the process is repeated. At iteration $t \in \mathbb{N}$, when $t-1$ candidates have been implemented and thus $X = (x_1, x_2, \ldots, x_{t-1})$, we let $r^t$ denote the ranking $\R(A,X) \in \mathcal{L}(C \setminus X)$ from which the DM can make a choice. 

We will sometimes make the assumption that the DM only ever implements alternatives that appear near the top of the ranking. In this \textit{depth-restricted setting}, we are given a natural number $h$ and we assume that $x_t \in r^t_{\le h}$ for all time steps~$t$. This setting models situations in which the DM does not have the resources (or the ability) to consider the whole ranking.

The straightforward ranking rule AV trivially translates to the dynamic setting: When a candidate is implemented, it is simply removed from the ranking; the order between the remaining candidates does not change. AV is used in all of the live Q\&A platforms mentioned in \Cref{sec:intro}. 

In the following, we propose dynamic variants of other ranking rules. For a more detailed description of these rules, including pseudocode and asymptotic runtime analysis, we refer to \Cref{app:rules}.

\paragraph{Dynamic seqPAV.}
For this straightforward dynamization of seqPAV, we modify the notion of marginal contribution to also take into account the satisfaction derived from previously implemented candidates:
\[ \mc_\text{dyn}(c) = \tsc(X \cup r_{\leq k} \cup  \{c\}) - \tsc(X \cup r_{\leq k}). \]
Dynamic seqPAV ranks candidates iteratively, adding in each round a candidate $c$ maximizing $\mc_\text{dyn}(c)$.
Note that $X$ is treated as a set here, as the order of elements in $X$ does not matter.

\paragraph{Dynamic Phragmén.}
Our first dynamization of sequential Phragmén works in two phases. As before, voters buy candidates and every candidate has a cost of 1 credit. Voters do not start with~$0$ credits, however; they may have an initial \textit{debt} due to previously implemented candidates they approve. The debts of voters are determined in the \textit{first phase}, which iterates through the sequence~$X$ (starting with $x_1$) and, for each implemented candidate $x_j \in X$, divides the cost of $1$ among the voters in $N_{x_j}$. More precisely, this assignment of debts is done in such a way that, in each iteration~$j$, the maximum total debt across all voters in $N_{x_j}$ is as small as possible. (The assignment of debts, therefore, mimics the assignment of loads in the load balancing formulation of sequential Phragmén.)
We let $d_i\ge 0$ denote the total debt of voter $i \in N$ resulting from this first phase. In the \textit{second phase}, we run sequential Phragmén to obtain the desired ranking of candidates in $C\setminus X$. At the beginning of this phase, each voter $i$ has a credit balance of $-d_i \le 0$. As in sequential Phragmén, voters continuously earn credits, and voters starting with debts can only participate in the purchase of a candidate once they have a positive balance. 

\medskip

These dynamic rules rank candidates in the same fashion as their non-dynamic counterparts, while taking the sequence $X$ of previously implemented candidates into account. (Note that the implementation \textit{order} matters for dynamic Phragmén, but not for dynamic seqPAV.) \hspace{0.1em} In particular, both dynamic rules coincide with their non-dynamic counterpart when $X=()$. 
Moreover, the ranking among the remaining candidates does not change whenever the top-ranked candidate is implemented: if $r^t=(r_1, r_2, r_3, \ldots)$ and $x_t = r_1$, then $r^{t+1}=(r_2, r_3, \ldots)$.

\medskip

We also consider two ``myopic'' dynamic ranking rules.

\paragraph{Myopic seqPAV.}
In this myopic dynamization of seqPAV, we compute the marginal contribution of each candidate $c \in C\setminus X$ only with respect to the set $X$ of previously implemented candidates, i.e., 
$\mc_\text{myopic}(c) = \tsc(X \cup \{c\}) - \tsc(X)$.
Then, we simply rank those candidates according to decreasing $\mc_\text{myopic}(c)$-value.

\paragraph{Myopic Phragmén.}
In this myopic dynamization of sequential Phragmén, we first run the first phase of dynamic Phragmén in order to determine the debts $\{d_i\}_{i \in N}$ of voters. 
Then, for each candidate $c \in C \setminus X$, we compute the voter debts that would result from adding candidate $c$ to~$X$ (and running the first phase for one more iteration). 
Let the debts induced by candidate $c$ be $\{d_i^c\}_{i \in N}$. 
Myopic Phragmén ranks the candidates in $C \setminus X$ according to increasing 
$\max_{i \in N_{c}} d^c_i$,
breaking ties according to the second highest debt and so on. 

\medskip

Intuitively, myopic seqPAV and myopic Phragmén rank candidates according to their suitability of being the next implemented candidate. In contrast to dynamic seqPAV and dynamic Phragmén, this way of comparing candidates does not lead to rankings that are representative by themselves. In particular, both myopic rules coincide with AV when $X=()$.

\medskip 
\begin{figure}
    \centering
    \begin{tikzpicture}
        [blue_cand/.style={rectangle, fill=blue!15, minimum width=1.7em}, 
        red_cand/.style={rectangle, fill=red!15, minimum width=1.7em}, 
        green_cand/.style={rectangle, fill=ForestGreen!20, minimum width=1.7em},
        chosen/.style={}, cand/.style={}, refr/.style={}, rank/.style={}]

        \node[rank] (rank) {};
        \node[blue_cand] (cand3) [right=0.5em of rank] {$b$};
        \node[blue_cand] (cand1) [above=1.35em of cand3] {$a$};
        \node[red_cand] (cand2) [above=0.1em of cand3] {$c$};
        \node[red_cand] (cand4) [below=0.11em of cand3] {$d$};
        \node[green_cand] (cand5) [below=1.6em of cand3] {$e$};
        \node[chosen] (chosen) [right=0.1em of cand3] {\checkmark};
        \node[refr] (dyn) [below right=1.5cm and 2cm of rank] {\small dynamic rules};
        \node[refr] (ref1) [left=0.3em of cand1.north west] {};
        \node[refr] (ref2) [right=0.3em of cand1.north east] {};
        \draw[very thick] (ref1) to (ref2);
        \draw [decorate,decoration={brace,amplitude=5pt}] (cand5.210) -- (cand1.150) node [black,midway,xshift=-1.5em] {$r^1$};
        
        \node[refr] (refA) [above right=1.2cm and 6.5cm of rank] {};
        \node[refr] (refB) [below right=1.5cm and 6.5cm of rank] {};
        \draw[thick] (refA) to (refB);
        
        \begin{scope}[xshift=2.4cm, yshift=-0.7em]
        \node[rank] (rank) {$\phantom{r^2}$};
        \node[blue_cand] (cand3) [above right=-0.75em and 0.5em of rank] {$a$};
        \node[blue_cand] (cand1) [above=1.35em of cand3] {$b$};
        \node[red_cand] (cand2) [above=0.1em of cand3] {$c$};
        \node[red_cand] (cand4) [below=0.11em of cand3] {$d$};
        \node[green_cand] (cand5) [below=1.6em of cand3] {$e$};
        \node[chosen] (chosen) [right=0.1em of cand4] {\checkmark};
        \node[refr] (ref1) [left=0.3em of cand1.south west] {};
        \node[refr] (ref2) [right=0.3em of cand1.south east] {};
        \draw[very thick] (ref1) to (ref2);
        \path [decorate,decoration={brace,amplitude=5pt}, thick] (cand1.south west) -- (cand1.north west)node [black,midway,xshift=-1.5em] {$X^2$};
        \draw [decorate,decoration={brace,amplitude=5pt}] (cand5.210) -- (cand2.150)node [black,midway,xshift=-1.5em] {$r^2$};
        \end{scope}
        
        \begin{scope}[xshift=4.8cm, yshift=-1em]
        \node[rank] (rank) {$\phantom{r^2}$};
        \node[blue_cand] (cand3) [above right=-0.75em and 0.5em of rank] {$a$};
        \node[blue_cand] (cand1) [above=1.6em of cand3] {$b$};
        \node[red_cand] (cand2) [above=0.1em of cand3] {$d$};
        \node[red_cand] (cand4) [below=0.11em of cand3] {$c$};
        \node[green_cand] (cand5) [below=1.35em of cand3] {$e$};
        \node[refr] (ref1) [left=0.3em of cand2.south west] {};
        \node[refr] (ref2) [right=0.3em of cand2.south east] {};
        \draw[very thick] (ref1) to (ref2);
        \draw [decorate,decoration={brace,amplitude=5pt}] (cand2.210) -- (cand1.150)node [black,midway,xshift=-1.5em] {$X^3$};
        \draw [decorate,decoration={brace,amplitude=5pt}] (cand5.210) -- (cand3.150)node [black,midway,xshift=-1.5em] {$r^3$};
        \end{scope}
        
        \begin{scope}[xshift=7.2cm]
        \node[rank] (rank) {$\phantom{r^1 =}$};
        \node[red_cand] (cand3) [right=0.5em of rank] {$c$};
        \node[blue_cand] (cand1) [above=1.6em of cand3] {$a$};
        \node[blue_cand] (cand2) [above=0.1em of cand3] {$b$};
        \node[red_cand] (cand4) [below=0.11em of cand3] {$d$};
        \node[green_cand] (cand5) [below=1.6em of cand3] {$e$};
        \node[chosen] (chosen) [right=0.1em of cand2] {\checkmark};
        \node[refr] (dyn) [below right=1.35cm and 2cm of rank] {\small myopic rules};
        \node[refr] (ref1) [left=0.3em of cand1.north west] {};
        \node[refr] (ref2) [right=0.3em of cand1.north east] {};
        \draw[very thick] (ref1) to (ref2);
        \draw [decorate,decoration={brace,amplitude=5pt}] (cand5.210) -- (cand1.150)node [black,midway,xshift=-1.5em] {$r^1$};
        \end{scope}
        
        \begin{scope}[xshift=9.6cm, yshift=-0.7em]
        \node[rank] (rank) {$\phantom{r^2 =}$};
        \node[red_cand] (cand3) [above right=-0.95em and 0.5em of rank] {$d$};
        \node[blue_cand] (cand1) [above=1.35em of cand3] {$b$};
        \node[red_cand] (cand2) [above=0.1em of cand3] {$c$};
        \node[blue_cand] (cand4) [below=0.11em of cand3] {$a$};
        \node[green_cand] (cand5) [below=1.35em of cand3] {$e$};
        \node[chosen] (chosen) [right=0.1em of cand3] {\checkmark};
        \node[refr] (ref1) [left=0.3em of cand1.south west] {};
        \node[refr] (ref2) [right=0.3em of cand1.south east] {};
        \draw[very thick] (ref1) to (ref2);
        \path [decorate,decoration={brace,amplitude=5pt}, thick, white] (cand1.south west) -- (cand1.north west)node [black,midway,xshift=-1.5em] {$X^2$};
        \draw [decorate,decoration={brace,amplitude=5pt}] (cand5.210) -- (cand2.150)node [black,midway,xshift=-1.5em] {$r^2$};
        \end{scope}
        
        \begin{scope}[xshift=12cm, yshift=-1em]
        \node[rank] (rank) {$\phantom{r^2}$};
        \node[blue_cand] (cand3) [above right=-0.75em and 0.5em of rank] {$a$};
        \node[blue_cand] (cand1) [above=1.6em of cand3] {$b$};
        \node[red_cand] (cand2) [above=0.1em of cand3] {$d$};
        \node[red_cand] (cand4) [below=0.11em of cand3] {$c$};
        \node[green_cand] (cand5) [below=1.35em of cand3] {$e$};
        \node[refr] (ref1) [left=0.3em of cand2.south west] {};
        \node[refr] (ref2) [right=0.3em of cand2.south east] {};
        \draw[very thick] (ref1) to (ref2);
        \draw [decorate,decoration={brace,amplitude=5pt}] (cand2.210) -- (cand1.150)node [black,midway,xshift=-1.5em] {$X^3$};
        \draw [decorate,decoration={brace,amplitude=5pt}] (cand5.210) -- (cand3.150)node [black,midway,xshift=-1.5em] {$r^3$};
        \end{scope}
        \end{tikzpicture}
    \caption{Rankings discussed in \Cref{ex:initial_ex}. Candidates approved by the three voter groups are marked in blue, red, and green, respectively. The rankings produced by the dynamic ranking rules are depicted on the left, the ones produced by the myopic rules on the right. The candidate that is chosen by the DM is marked with ``\checkmark'' and appears in the sequence of implemented candidates $X$ in the next iteration.
    In each iteration $t$, the thick line separates the sequence $X^t$ of implemented candidates (above the line) from the ranking over the remaining candidates in $C\setminus X^t$ (below the line).}
    \label{fig:initial_ex}
\end{figure}
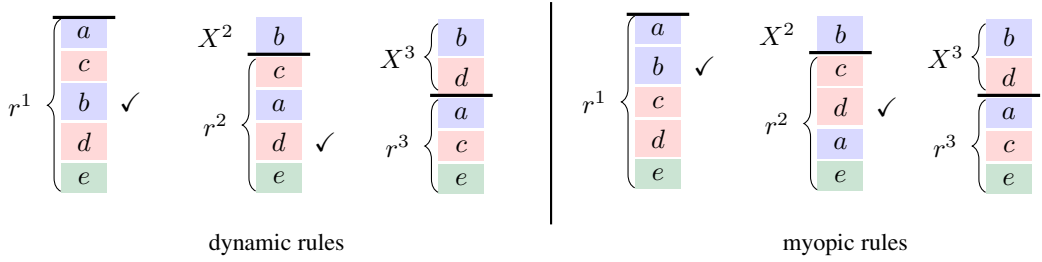
We illustrate these rules with a simple example. The rankings discussed in this example are depicted in \Cref{fig:initial_ex}.
\begin{example}\label{ex:initial_ex}
Let $C=\Set{a,b,c,d,e}$ and assume alphabetic tiebreaking. Consider a set of $9$ voters with the following approval sets:
\[5 \times \Set{a,b}, \qquad 3 \times \Set{c,d}, \qquad 1 \times \Set{e} \text. \]
Let $V$ denote the group consisting of the $5$ $\{a,b\}$-voters and $V'$ the group consisting of the $3$ $\{c,d\}$-voters.
First, consider dynamic seqPAV and dynamic Phragmén. 
In the first iteration, 
both rules output $r^1 = (a,c,b,d,e)$, effectively alternating between candidates supported by voter groups $V$ and $V'$. Let us assume that the DM first implements candidate $x_1=b$, i.e., $X^2=(b)$. Then, the two rules output $r^2=(c,a,d,e)$. If the DM implements candidate $x_2=d$ next (and thus $X^3=(b,d)$), both rules output $r^2=(a,c,e)$. 

Next, consider myopic seqPAV and myopic Phragmén. 
In the first iteration, both rules (and AV) rank the candidates according to their approval scores: $r^1 = (a,b,c,d,e)$. 
After the implementation of $b$, both rules output $r^2 = (c,d,a,e)$, which differs from the AV ranking $r^2=(a,c,d,e)$. If the DM then implements candidate $x_2=d$, the two rules output $r^2=(a,c,e)$. 
\end{example}

In this example, all of our ranking rules demote candidate $a$ in $r^2$ because voter group $V$ is already (partially) satisfied with $X^2=(b)$. The myopic rules even rank \textit{both} $c$ and $d$ higher than $a$ in~$r^2$, since implementing either $c$ or $d$ would yield a more proportional sequence $X$ than implementing $a$ would.

All presented ranking rules can be computed in polynomial time; see \Cref{app:rules} for details.

\section{Monotonicity of Voter Satisfaction}
\label{sec:impl_mono}

We start our analysis of dynamic ranking rules by considering the satisfaction of voters during the sequential selection process. In doing so, we assume that voters derive satisfaction not only from implemented candidates they approve, but also---possibly to a lesser extent---from approved candidates appearing near the top of the ranking: high positions in the ranking come with increased attention (and, presumably, high selection probabilities in future iterations) for the respective candidates.  
In particular, improved ranking positions of supported candidates can be viewed as a kind of compensation for (groups of) voters who are not (yet) well-represented by the implemented alternatives. 
To make this concrete, consider an iteration~$t$, where the DM is confronted with ranking $r^t$ and chooses to implement candidate $x_t$.  
Following the logic outlined above, it might be natural to expect that voters \textit{not} approving $x_t$ (or, more precisely, the candidates approved by these voters) should get a ``boost'' in the ranking. At the very least, it seems reasonable to expect that the satisfaction of such voters with the new ranking $r^{t+1}$ is at least as high as with the old ranking $r^t$. 
AV trivially satisfies this property, which we informally refer to as \textit{satisfaction \mbox{monotonicity}}.
Somewhat surprisingly, however, the following simple example demonstrates that this intuitive monotonicity notion is not achievable for dynamic ranking rules that satisfy a minimal degree of representativeness. 
The rankings discussed in this example are depicted in \Cref{fig:simple_mono}.\footnote{\Cref{ex:simple_mono} can be turned into an impossibility result: Every dynamic ranking rule that \textit{(i)} ranks the approval winner at the top in the first iteration and \textit{(ii)} gives priority to less satisfied voter groups fails satisfaction monotonicity.}

\begin{example}\label{ex:simple_mono}
Consider the following profile with 7 voters: 
\[ 1 \times \Set{a}, \qquad 3 \times \Set{b}, \qquad 3 \times \Set{a, c}. \]
All rules considered in this paper rank the approval winner~$a$ first in $r^1$.
If the DM chooses to implement candidate~$x_1 = c$,  
all of our rules---except AV---output $r^2 = (b,a)$ in the second iteration. 
Intuitively, the rules give more voting power to the 3 supporters of $b$ (all of which are unrepresented by $c$) than to the 4 supporters of $a$ (3 of which are already partially represented). 
Observe that the satisfaction of the voter approving $a$ decreases when going from $r^1$ to $r^2$, despite the fact that this voter does not approve the candidate being implemented.  
\end{example}

The following definition is motivated by the question whether monotonicity failures can be prevented by moving to the depth-restricted setting and putting lower bounds on the size of voter groups for which monotonicity should hold. 

\begin{definition}
For $h \geq 1$ and $\alpha \in (0,1]$, a dynamic ranking rule satisfies \emph{$(h,\alpha)$-monotonicity} if, for all profiles and all groups of voters $V \subseteq N$ of size $|V| \geq \alpha\cdot |N|$, the following holds for every iteration $t$:
\[ \text{If } x_t \notin \bigcup_{i \in V} A_i \text{, then } \avg_V(r^{t+1}_{\leq h}) \geq \avg_V(r^t_{\leq h}). \]
\end{definition}

That is, $(h,\alpha)$-monotonicity requires that satisfaction monotonicity holds for groups that make up at least an $\alpha$-fraction of the electorate, and when measuring satisfaction with respect to the first~$h$ positions in a ranking.

AV trivially satisfies $(h,\alpha)$-monotonicity for all~$h$ and all~$\alpha$. On the other hand, all other considered rules violate this notion unless we consider rather large groups of voters. 

\begin{figure}[t]
    \centering
    \begin{tikzpicture}
        [blue_cand/.style={rectangle, fill=blue!15, minimum width=1.7em}, 
        red_cand/.style={rectangle, fill=red!15, minimum width=1.7em}, 
        green_cand/.style={rectangle, fill=green!15, minimum size=1.7em},
        chosen/.style={}, cand/.style={minimum width=1.7em}, refr/.style={}, rank/.style={}]

        \node[rank] (rank) {$\phantom{r^1 =}$};
        \node[cand] (cand2) [right=0.5em of rank] {$b$};
        \node[red_cand] (cand1) [above=0.1em of cand2] {$a$};
        \node[cand] (cand3) [below=0.11em of cand2] {$c$};
        \node[chosen] (chosen) [right=0.1em of cand3] {\checkmark};
        \node[refr] (ref1) [left=0.3em of cand1.north west] {};
        \node[refr] (ref2) [right=0.3em of cand1.north east] {};
        \draw[very thick] (ref1) to (ref2);
        \draw [decorate,decoration={brace,amplitude=5pt}] (cand3.210) -- (cand1.150)node [black,midway,xshift=-1.5em] {$r^1$};
        
        \begin{scope}[xshift=2.5cm]
        \node[rank] (rank) {\phantom{or}};
        \node[cand] (cand2) [right=1.5em of rank] {$c$};
        \node[red_cand] (cand1) [above=0.1em of cand2] {$a$};
        \node[cand] (cand3) [below=0.11em of cand2] {$b$};
        \node[chosen] (chosen) [right=0.1em of cand2] {\checkmark};
        \node[refr] (ref1) [left=0.3em of cand1.north west] {};
        \node[refr] (ref2) [right=0.3em of cand1.north east] {};
        \draw[very thick] (ref1) to (ref2);
        \draw [decorate,decoration={brace,amplitude=5pt}] (cand3.210) -- (cand1.150) node (r1right) [black,midway,xshift=-1.5em] {$r^1$};
        \node[rank] (or) [below left=-0.5em and 0em of r1right.west] {or};
        \end{scope}
        
        \begin{scope}[xshift=6cm, yshift=-0.5em]
        \node[rank] (rank) {$\phantom{r^2 =}$};
        \node[cand] (cand2) [above right=-0.75em and 0.5em of rank] {$b$};
        \node[cand] (cand1) [above=0.1em of cand2] {$c$};
        \node[red_cand] (cand3) [below=0.11em of cand2] {$a$};
        \node[refr] (ref1) [left=0.3em of cand1.south west] {};
        \node[refr] (ref2) [right=0.3em of cand1.south east] {};
        \draw[very thick] (ref1) to (ref2);
        \draw [decorate,decoration={brace,amplitude=5pt}] (cand3.210) -- (cand2.150)node [black,midway,xshift=-1.5em] {$r^2$};
        \path [decorate,decoration={brace,amplitude=5pt}] (cand1.210) -- (cand1.150)node [black,midway,xshift=-1.5em] {$X^2$};
        \end{scope}
        
        \end{tikzpicture}
    \caption{Rankings discussed in \Cref{ex:simple_mono}. 
    All rules considered here either output $r^1=(a,b,c)$ or $r^1=(a,c,b)$. If the DM chooses to implement candidate~$x_1 = c$, all of the rules---except AV---output $r^2 = (b,a)$ in the second iteration. This violates an intuitive understanding of \emph{monotonicity} for the voter with ballot $\Set{a}$.}
    \label{fig:simple_mono}
\end{figure}

\begin{proposition}\label{thm:implmono}
Consider the depth-restricted setting for some $h\geq 3$. 
Then, dynamic seqPAV and dynamic Phragmén fail to satisfy $(h,\alpha)$-monotonicity for all $\alpha < \frac{6}{2h+5}$.  
Furthermore, myopic seqPAV and myopic Phragmén fail to satisfy $(h,\alpha)$-monotonicity for all $\alpha < \frac{1}{h}$.
\end{proposition}

\begin{proof}[Proof sketch]
We first consider dynamic seqPAV and then extend the argument to the other rules. Let $h=3$ and $j = 6\cdot y$ for some $y \in \mathbb{N}$ and consider the profile given by
\begin{align*}
    2 &\times \Set{a}, ~ 15 \times \Set{a,b}, ~ \left(\frac{j}{2} + 6\right) \times \Set{b}, ~ 10 \times \Set{c}, \\
    10 &\times \Set{d}, ~ (j + 6) \times \Set{a,c,d}, ~ \left(\frac{j}{3} + 12 \right) \times \Set{e}.
\end{align*}
Initially, dynamic seqPAV outputs $r^1=(a,b,c,d,e)$. After the DM implements $x_1 = b$, dynamic seqPAV outputs $r^2=(c,d,e,a)$ (where the ordering of $c$ and $d$ might depend on tie-breaking). Now consider the voter group~$V$ consisting of all the voters with approval sets $\Set{a}$ or $\Set{a,c,d}$. In the first iteration, the average satisfaction of this group is $\avg_V(r^1_{\leq 3}) = \frac{1}{8+j} \cdot (14 + 2j)$; in the second iteration, it is only $\avg_V(r^2_{\leq 3}) = \frac{1}{8+j} \cdot (12 + 2j)$. For $j \rightarrow \infty$, this group of voters makes up nearly $6/11$ of the electorate (and the rankings remain unchanged).
To extend this example to the case of $h > 3$, we can clone alternative $e$ and all its supporters. 

The dynamic version of Phragmén's rule violates monotonicity on the same examples (see \Cref{app:implmono} for details).
After replacing the two terms $j/3$ and $j/2$ in the above profile each with $j$, and for each copy of $e$ adding 4 voters approving that candidate only, the same holds for myopic seqPAV and myopic Phragmén.
\end{proof}

The monotonicity requirement can be weakened further by only requiring satisfaction monotonicity in cases in which the implemented candidate is never co-approved with any candidate that is approved by a member of the group under consideration (i.e., there is no $c \in \bigcup_{k \in V} A_k$ with $\{c,x_t\}\subseteq A_i$ for some $i \in N$). Both myopic rules satisfy this weak implementation monotonicity, whereas the two dynamic versions fail it. For a thorough discussion of this weaker version of monotonicity, we refer to \Cref{app:implmono}.

Despite the negative results in this section, we rarely found monotonicity violations of any kind in our experiments (see \Cref{sec:experiments}).

\section{Proportional Representation}\label{sec:proportionality}

We now turn to analyzing the proportional representativeness that is provided by our dynamic ranking rules. The following two sections capture different perspectives on representation, focusing on the representativeness of the ranking $r^t$ at any given iteration $t$ (\Cref{sec:exp_guarantee}) and on the representativeness of the set $X$ of implemented candidates (\Cref{sec:sel_guarantee}).

\subsection{Proportionality of Rankings}\label{sec:exp_guarantee}

In certain applications of dynamic ranking rules, such as the live Q\&A platforms mentioned in the introduction, it is desirable for the ranking $r^t$ to provide a representative overview of the opinions of the voters at any given iteration $t$. In this section, we prove proportionality guarantees that are satisfied by ranking $r^t$ for any fixed iteration $t$.

Measures for the proportionality of a ranking have been proposed by \citet{SLB+17a}. In particular, \textit{$\kappa$-group representation} measures, informally speaking, how far down in the ranking a group of voters needs to go in order to obtain a given amount of satisfaction. 
In order to adapt the notion of $\kappa$-group representation to the dynamic ranking setting, we need the following notation. For iteration $t$, let $X^t=\{x_1, \ldots, x_{t-1}\}$ denote the set of candidates implemented in the first $t-1$ rounds and, for a group $V\subseteq N$ of voters, let $\lambda^t(V) = |\bigcap_{i\in V} A_i \setminus X^t|$ denote the cohesiveness of~$V$ with respect to the remaining candidates $C \setminus X^t$.

\begin{definition}[Group representation]\label{def:grouprep}
Let $\gr(\alpha,\lambda)$ be a function from $((0,1] \cap \mathbb{Q}) \times \mathbb{N})$ to $\mathbb{N}$.
A dynamic ranking rule satisfies \emph{$\kappa$-group representation} if the following holds for all profiles $A$, groups of voters $V\subseteq N$, rational numbers $\alpha \in (0,1]$, and integers $\lambda, t \le |C|$: 

If $|V| \geq \alpha \cdot n$ and $\lambda^t(V) \geq \lambda$, then
$\avg_V(r^t_{\leq \gr(\alpha,\lambda)}) \geq \lambda$.
\end{definition}

In words: If a group $V$ of voters makes up an $\alpha$-fraction of the electorate and has at least $\lambda$ commonly approved candidates remaining at iteration $t$, then this group derives an average satisfaction of at least $\lambda$ from the candidates ranked in the top $\kappa(\alpha,\lambda)$ positions of ranking $r^t$.\footnote{A natural lower bound for $\kappa(\alpha,\lambda)$ is given by $\lceil \lambda / \alpha \rceil$.
Note that the $\kappa$ functions used in this section not only depend on $\alpha$ and $\lambda$, but also on the set $V$ and on the sequence $X$ of previously implemented candidates. In an attempt to simplify notation, we decided to not make this dependencies explicit in \Cref{def:grouprep}.}

Our first result in this section is for dynamic Phragmén.
We recall that $d_i$ denotes the initial debt of voter $i$ at the end of the first phase of the method, and let $d^V_\text{avg}=\frac{1}{|V|} \sum_{i \in V} d_i$ denote the average debt of voters in $V$.

\begin{theorem}\label{thm:gr_dphrag}
Dynamic Phragmén satisfies $\kappa$-group representation for 
    \[ \kappa(\alpha, \lambda) = \left\lceil \frac{2(\lambda + m + 1) + s\cdot|V|}{\alpha}\right\rceil \text, \] 
    where $m = |\bigcup_{i\in V} A_i \cap X|$ and $s = \sum_{i\in V} (d_i - d^V_\text{avg})^2$. 
\end{theorem}
Observe that this function is increasing both in the number $m$ of already implemented candidates that are approved by some voter in $V$ and in the variance $s$ of debts of voters in~$V$. For the special case $X=()$, \Cref{thm:gr_dphrag} implies a group representation of $\left\lceil \frac{2\lambda + 2}{\alpha}\right\rceil$ for (non-dynamic) sequential Phragmén. For $\lambda \ge 2$, this is an improvement over the $\kappa$-group representation bound of $\left\lceil \frac{5\lambda}{\alpha^2} + \frac{1}{\alpha}\right\rceil$ proved by \citet{SLB+17a}.

The proof of \Cref{thm:gr_dphrag} employs the notion of  \emph{proportionality degree} \citep{Skow18a}. In particular, we first prove a bound on the proportionality degree of dynamic Phragmén, using a potential function approach that is similar to the one used by \citet{Skow18a} for the non-dynamic setting. Then, we establish a relationship between the proportionality degree and group representation, and use it to translate the bound on the former into a bound on the latter. 
For details, see \Cref{app:exposure}.

For dynamic seqPAV we prove the following generalisation of Theorem 3 by \citet{SLB+17a}; the latter theorem corresponds to the special case $X = ()$.

\begin{theorem}\label{thm:gr_dseqpav}
Dynamic seqPAV satisfies $\kappa$-group representation for 
\[ \kappa(\alpha, \lambda) = \left\lceil \frac{2(\lambda + 1 + \avg_V(X))^2}{\alpha^2} \right\rceil. \]
\end{theorem}

AV does not perform any different in the dynamic ranking setting compared to the non-dynamic one. Thus, it satisfies the same bounds on group representation as those stated in Theorem 2 by \citet{SLB+17a}. Since myopic seqPAV and myopic Phragmén both agree with AV in the case \mbox{$X=()$}, the same bounds hold for these two rules.  

\begin{proposition}\label{thm:gr_lazy}
Myopic seqPAV and myopic Phragmén fail $\kappa$-group representation for \[\kappa(\alpha, \lambda) \leq \left\lceil \frac{\lambda \cdot \alpha}{2\alpha -1} \right\rceil -1\] and for all functions $\kappa(\alpha, \lambda)$ if $\alpha \leq \frac{m+1}{m+2}$, where $m = |\bigcup_{i\in V} A_i \cap X|$.
\end{proposition}

\subsection{Proportionality of Implemented Candidates}
\label{sec:sel_guarantee}

In this section, we study worst-case bounds on the representativeness of the set $X$ of implemented candidates. 
Clearly, no non-trivial bounds are obtainable without restricting the selection behavior of an adversarial DM. Therefore, we will make the following two assumptions throughout this section:
\begin{enumerate}
    \item[(A1)] The DM is depth-restricted and always implements a candidate from the top $h$ positions of the ranking. 
    \item[(A2)] Every candidate $c \in C$ has sufficiently\footnote{Given an upper bound $T$ on the number of iterations, $h+T-1$ clones suffice (as at least $h$
    clones will always remain in the ranking).} many ``clones,'' i.e., candidates $c'$ with identical supporter set $N_{c'} = N_c$. 
\end{enumerate}
Assumptions (A1) and (A2) together ensure that the DM can be forced to implement a candidate approved by a voter, by populating the top $h$ positions exclusively with such candidates. 
Arguably the most natural way to ensure (A2) is to assume that we are in the \textit{party-approval} setting \citep{BGP+19a}, where candidates are interpreted as parties and can be selected arbitrarily often. In the motivating example of live Q\&A platforms, party-approval preferences could result from assigning attributes to questions and eliciting participants' approval preferences over attributes.

Recall that $X^{t+1}$ denotes the set containing the implemented candidates from the first $t$ rounds. 
The following property is a natural adaption of the well-studied proportionality axiom \textit{proportional justified representation (PJR)} \citep{SFF+17a}.

\begin{definition}\label{def:PJS}
A dynamic ranking rule satisfies \textit{proportional justified selection (PJS)} if the following holds for all $t,\ell \in \mathbb{N}$ and for all groups $V \subseteq N$ of voters:
If $|V| \geq \frac{\ell}{t} \cdot |N|$ and $|\bigcap_{i \in V} A_i|\ge \ell$, then 
$|X^{t+1} \cap \bigcup_{i\in V} A_i | \geq \ell$.
\end{definition}

A weaker version of this axiom is obtained by fixing $\ell=1$; in analogy to a well-known notion due to \citet{ABC+16a}, we refer to the resulting property as \emph{justified selection (JS)}.

We prove the following theorem by interpreting the set $X^{t+1}$ of implemented alternatives as a committee. 

\begin{theorem}
Under assumptions (A1) and (A2), myopic Phragmén satisfies PJS.
\end{theorem}

\begin{proof} 
First observe that myopic Phragmén always ranks clones consecutively. Due to (A2), there are always at least $h$ clones of each candidate, so that the first $h$ position of each ranking $r^{t'}$ (where $t'\le t$) will be occupied by a set of candidates that are all clones of each other. Due to (A1), the DM selects a candidate from this top-ranked clone set in each iteration. 
Now consider the set $X^{t+1}=\{x_1, \ldots, x_t\}$ of implemented candidates. Since the assignment of debts under myopic Phragmén mimics the distribution of loads under sequential Phragmén, this set consists precisely of the first~$t$ candidates that sequential Phragmén selects on the same instance. Since sequential Phragmén satisfies PJR \citep{BFJL16a}, it follows that myopic Phragmén satisfies PJS.
\end{proof}

Analogously, we can translate a representation guarantee for sequential PAV \citep{SFF+17a} into a guarantee for myopic seqPAV.

\begin{proposition}
Under assumptions (A1) and (A2), myopic seqPAV satisfies JS for $t\leq 5$. 

\end{proposition}

\newcommand{\dotcup}{\mathbin{\dot{\cup}}}

\noindent
Similar positive results are not possible for the other rules. 
To see this, consider the following example, which is consistent with assumptions (A1) and (A2). The rankings discussed in this example are depicted in \Cref{fig:dyn_fail_js}.
\begin{example}\label{ex:dyn_fail_js}
Let $N=V \cup G$  be the electorate consisting of two disjoint groups of voters of equal size, i.e., $V \cap G = \emptyset$ and $|V| = |G|$. Now assume that each voter in $V$ approves of all candidates in $\Set{ a_1, a_2, \ldots}$ and each voter in $G$ approves of all candidates in $\Set{b_1, b_2, \ldots}$. If we assume alphabetic tie-breaking between the parties then both dynamic rules will in the first iteration output the ranking $r^1 = (a_1, b_1, a_2, b_2, \ldots)$. If we set $h = 4$ then it is possible for an adversarial DM to implement 3 candidates supported by $G$ before in the fourth iteration we have $r^4_{\leq 4} = \Set{a_1,a_2,a_3,a_4}$. In particular we have $X^3 = (b_1, b_2)$ and $X^4 = (b_1, b_2, b_3)$.
\end{example}
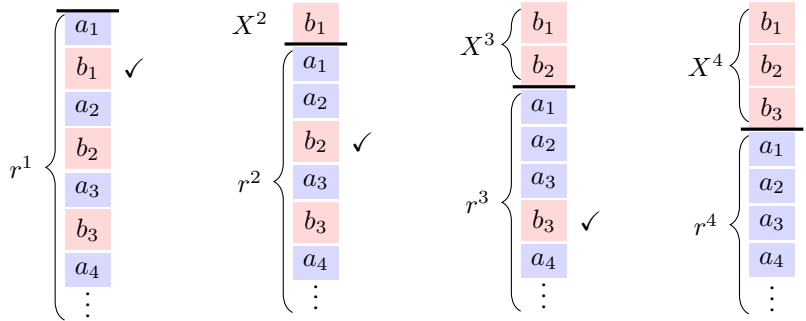
\begin{figure}
    \centering
    \begin{tikzpicture}
        [blue_cand/.style={rectangle, fill=blue!15, minimum width=1.7em}, 
        red_cand/.style={rectangle, fill=red!15, minimum width=1.7em}, 
        green_cand/.style={rectangle, fill=green!15, minimum size=1.5em},
        chosen/.style={}, cand/.style={minimum width=1.7em}, refr/.style={}, rank/.style={}]

        \node[rank] (rank) {$\phantom{r^1 =}$};
        \node[red_cand] (cand4) [right=0.5em of rank] {$b_2$};
        \node[blue_cand] (cand1) [above=3.15em of cand4] {$a_1$};
        \node[red_cand] (cand2) [above=1.5em of cand4] {$b_1$};
        \node[blue_cand] (cand3) [above=0.1em of cand4] {$a_2$};
        \node[blue_cand] (cand5) [below=0.11em of cand4] {$a_3$};
        \node[red_cand] (cand6) [below=1.5em of cand4] {$b_3$};
        \node[blue_cand] (cand7) [below=3.15em of cand4] {$a_4$};
        \node[cand] (cand8) [below=3.7em of cand4] {\vdots};
        \node[chosen] (chosen) [right=0.1em of cand2] {\checkmark};
        \draw [decorate,decoration={brace,amplitude=7pt}] (cand8.225) -- (cand1.150)node [black,midway,xshift=-1.6em] {$r^1$};
        \node[refr] (ref1) [left=0.3em of cand1.north west] {};
        \node[refr] (ref2) [right=0.3em of cand1.north east] {};
        \draw[very thick] (ref1) to (ref2);
        
        \begin{scope}[xshift=3cm, yshift=-0.5em]
        \node[rank] (rank) {$\phantom{r^2 =}$};
        \node[red_cand] (cand4) [above right=-0.75em and 0.5em of rank] {$b_2$};
        \node[red_cand] (cand1) [above=2.9em of cand4] {$b_1$};
        \node[blue_cand] (cand2) [above=1.5em of cand4] {$a_1$};
        \node[blue_cand] (cand3) [above=0.1em of cand4] {$a_2$};
        \node[blue_cand] (cand5) [below=0.11em of cand4] {$a_3$};
        \node[red_cand] (cand6) [below=1.5em of cand4] {$b_3$};
        \node[blue_cand] (cand7) [below=3.15em of cand4] {$a_4$};
        \node[cand] (cand8) [below=3.7em of cand4] {\vdots};
        \node[refr] (ref1) [left=0.3em of cand1.south west] {};
        \node[refr] (ref2) [right=0.3em of cand1.south east] {};
        \draw[very thick] (ref1) to (ref2);
        \node[chosen] (chosen) [right=0.1em of cand4] {\checkmark};
        \draw [decorate,decoration={brace,amplitude=7pt}] (cand8.225) -- (cand2.150)node [black,midway,xshift=-1.6em] {$r^2$};
        \path [decorate,decoration={brace,amplitude=10pt}] (cand1.210) -- (cand1.150)node [black,midway,xshift=-1.6em] {$X^2$};
        \end{scope}
        
        \begin{scope}[xshift=6cm, yshift=-1.5em]
        \node[rank] (rank) {$\phantom{r^3 =}$};
        \node[blue_cand] (cand4) [above right=0.3em and 0.5em of rank] {$a_2$};
        \node[red_cand] (cand1) [above=3.15em of cand4] {$b_1$};
        \node[red_cand] (cand2) [above=1.5em of cand4] {$b_2$};
        \node[blue_cand] (cand3) [above=0.1em of cand4] {$a_1$};
        \node[blue_cand] (cand5) [below=0.11em of cand4] {$a_3$};
        \node[red_cand] (cand6) [below=1.5em of cand4] {$b_3$};
        \node[blue_cand] (cand7) [below=3.15em of cand4] {$a_4$};
        \node[cand] (cand8) [below=3.7em of cand4] {\vdots};
        \node[refr] (ref1) [left=0.3em of cand2.south west] {};
        \node[refr] (ref2) [right=0.3em of cand2.south east] {};
        \draw[very thick] (ref1) to (ref2);
        \node[chosen] (chosen) [right=0.1em of cand6] {\checkmark};
        \draw [decorate,decoration={brace,amplitude=7pt}] (cand8.225) -- (cand3.150)node [black,midway,xshift=-1.6em] {$r^3$};
        \draw [decorate,decoration={brace,amplitude=7pt}] (cand2.210) -- (cand1.150)node [black,midway,xshift=-1.6em] {$X^3$};
        \end{scope}
        
        \begin{scope}[xshift=9cm, yshift=-2.3em]
        \node[rank] (rank) {$\phantom{r^4 =}$};
        \node[blue_cand] (cand4) [above right=0.9em and 0.5em of rank] {$a_1$};
        \node[red_cand] (cand1) [above=3.35em of cand4] {$b_1$};
        \node[red_cand] (cand2) [above=1.72em of cand4] {$b_2$};
        \node[red_cand] (cand3) [above=0.1em of cand4] {$b_3$};
        \node[blue_cand] (cand5) [below=0.11em of cand4] {$a_2$};
        \node[blue_cand] (cand6) [below=1.5em of cand4] {$a_3$};
        \node[blue_cand] (cand7) [below=2.9em of cand4] {$a_4$};
        \node[cand] (cand8) [below=3.6em of cand4] {\vdots};
        \node[refr] (ref1) [left=0.3em of cand3.south west] {};
        \node[refr] (ref2) [right=0.3em of cand3.south east] {};
        \draw[very thick] (ref1) to (ref2);
        \draw [decorate,decoration={brace,amplitude=7pt}] (cand8.225) -- (cand4.150)node [black,midway,xshift=-1.6em] {$r^4$};
        \draw [decorate,decoration={brace,amplitude=7pt}] (cand3.210) -- (cand1.150)node [black,midway,xshift=-1.6em] {$X^4$};
        \end{scope}
        \end{tikzpicture}
    \caption{Rankings discussed in \Cref{ex:dyn_fail_js}. 
    }
    \label{fig:dyn_fail_js}
\end{figure}

\begin{proposition}\label{thm:dyn_fail_JS}
Dynamic seqPAV and dynamic Phragmén fail to satisfy JS, even under assumptions (A1) and (A2) and for $t=2$.
\end{proposition}

\section{Experimental Evaluation}\label{sec:experiments}

In  order  to  better  understand  the behavior of the dynamic ranking rules considered in this paper, we conducted computational experiments using randomly generated approval profiles. Since we were mainly interested in the proportional representation of groups of voters with similar preferences, we generated profiles according to two probabilistic models that lead to polarized electorates with easily identifiable groups. 
We measured
(1) how the satisfaction of a voter group with the set of implemented candidates varies with the size of the group, and
(2) how the satisfaction of a voter group with the current ranking varies over time. 
\footnote{The code used to perform these experiments can be accessed via \url{https://git.tu-berlin.de/jonas.israel/dpr}.}

\paragraph{Setup.}
All of our profiles consist of 60 voters and 20 candidates, and the approval sets are generated according to two different models. 
In the \emph{blurred parties model}, we assign each voter and each candidate to one of \textit{two} parties. The size of the voter group $V$ associated with the first party varies over the experiments; the candidates are always divided equally. 
Each voter approves a candidate from their own party with 95\% probability and a candidate of the other party with 5\% probability. 
The \emph{spatial model} is an adaption of the 4-Gaussian model used by \citet{EFL+17a} for linear preferences. Voters and candidates correspond to points in the Euclidean plane and voters approve nearby candidates. There are \textit{three} parties with equidistant center locations, and candidates as well as voters get sampled as points around the party centers according to a normal distribution.  
We let the size of the voter group $V$ associated with the first party grow, and divide the remaining voters equally among the two remaining parties. There are 7 candidates associated with the first party. 

The selection behavior of the DM is modeled via Google click-through rates. In particular, the probability of selection decreases when going down the ranking. 
In \Cref{fig:experiments} we plot the averages of 100 generated elections. More details about the setup can be found in \Cref{app:experiments}. 

\begin{figure}
    \centering
    \scalebox{1.25}{ \hspace{-1.25em}
    \begin{subfigure}[b]{0.4\textwidth}   
        \centering 
        \includegraphics[width=\textwidth]{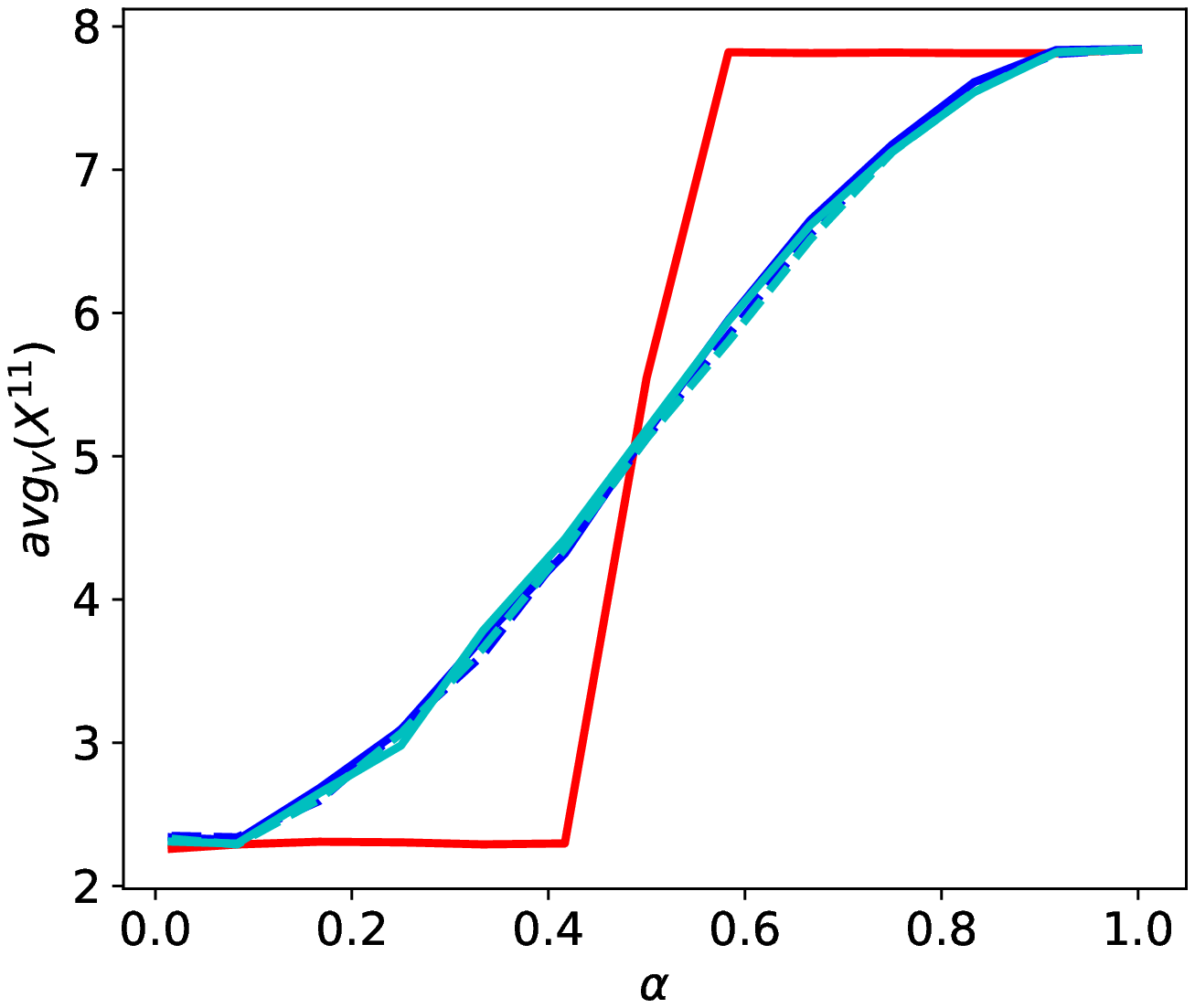}
    \end{subfigure}
    \begin{subfigure}[b]{0.4\textwidth}   
        \centering 
        \includegraphics[width=\textwidth]{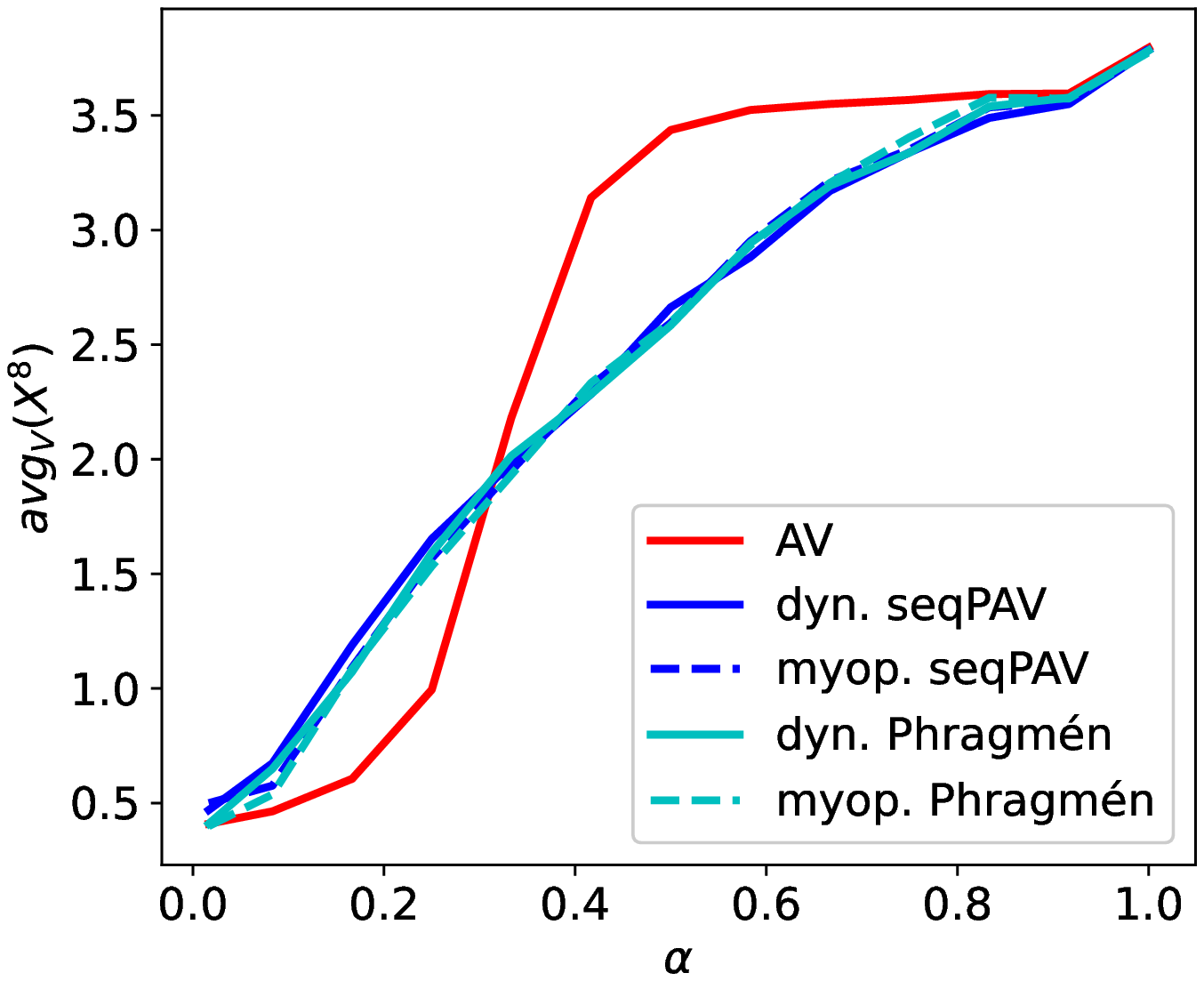}
    \end{subfigure}
    }
    \scalebox{1.3}{ \hspace{-1.25em}
    \begin{subfigure}[b]{0.4\textwidth}
        \centering
        \includegraphics[width=\textwidth]{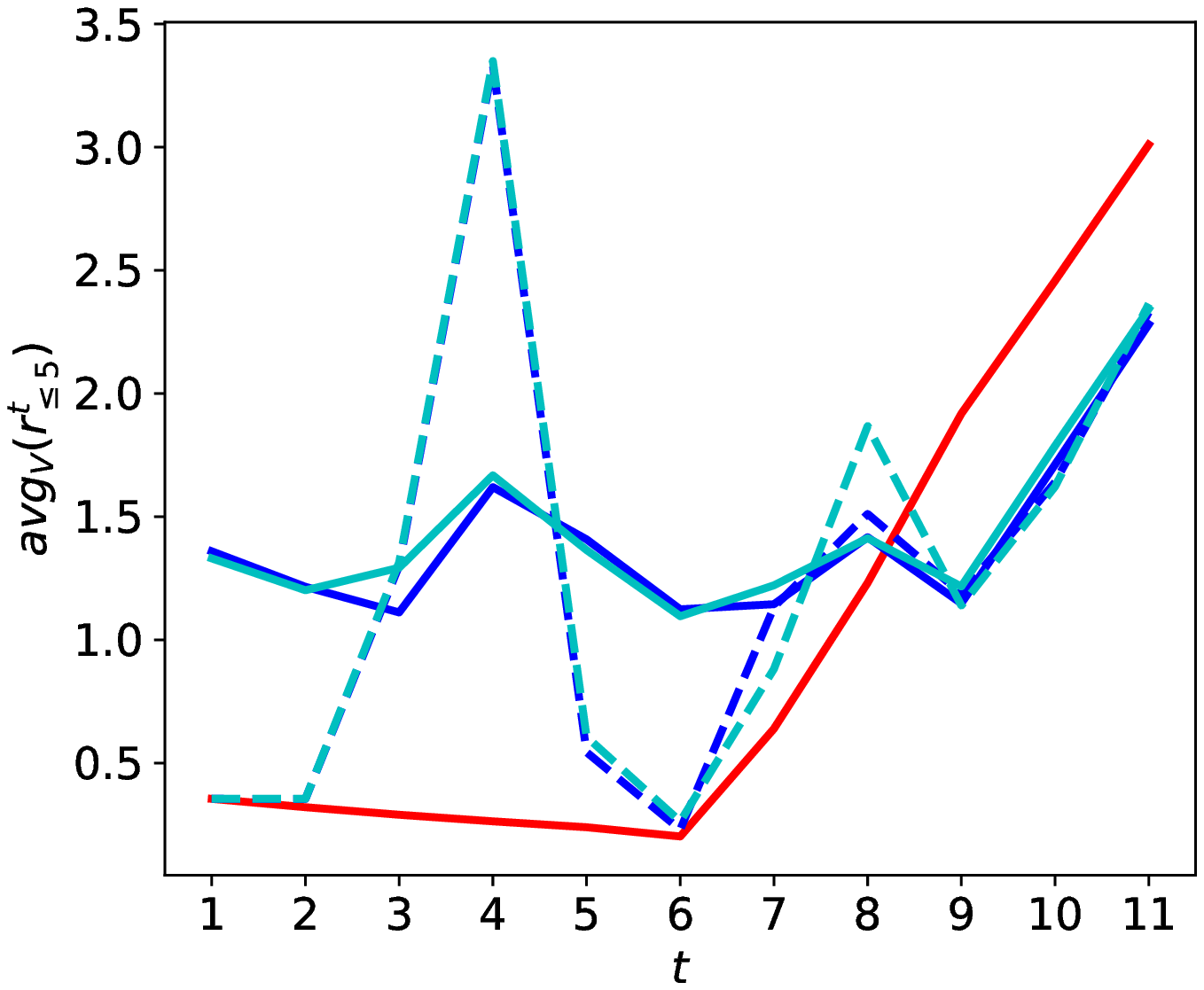}
    \end{subfigure}
    \begin{subfigure}[b]{0.4\textwidth}  
        \centering 
        \includegraphics[width=\textwidth]{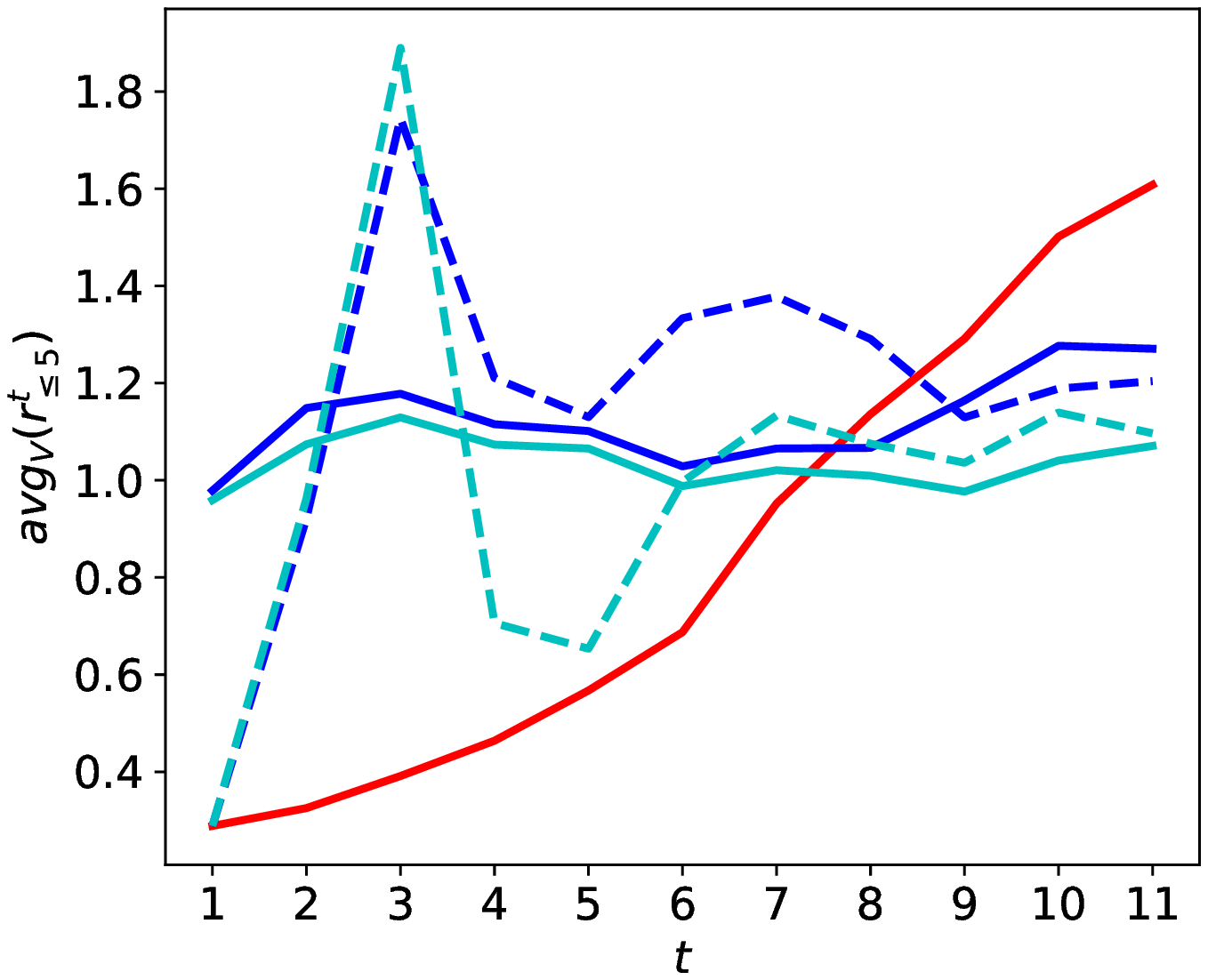}
    \end{subfigure}
    }
    \caption{Experimental results for the blurred parties model (left) and the spatial model (right). 
    The graphs in the first row show the average satisfaction of $V$ with the first $k$ implemented candidates, for relative group size $\alpha \in [0,1]$.
    The graphs in the second row show the average satisfaction of $V$ with $r^t_{\le 5}$, for $1 \le t \le 11$.} 
    \label{fig:experiments}
\end{figure}

\paragraph{Satisfaction with implemented candidates.}
We measure the average satisfaction of voter group $V\subseteq N$ with $X^{k+1}$, where~$k$ is the number of candidates associated with that group (i.e., $k=10$ for the blurred parties model and $k=7$ for the spatial model). We plot this value against the relative size $\alpha = |V|/|N|$ of the group $V$. The graphs in the first row of \Cref{fig:experiments} show that for both models, AV is not proportional: $\avg_V(X^{k+1})$ starts out very low and only jumps up as soon as $V$ becomes the biggest group (which happens at $\alpha = 1/2$ and $\alpha = 1/3$, respectively). 
In other words, AV underrepresents minorities and overrepresents majorities.
The performance of the other four rules are indistinguishable, as all yield proportionally increasing satisfaction values.

\paragraph{Satisfaction with rankings.}
The graphs in the second row of \Cref{fig:experiments} depict the average satisfaction of a group $V$ of size $\alpha = 1/4$ with the first 5 candidates of the ranking over the first 11 (respectively 8) iterations. Again, AV behaves poorly, as it gives satisfaction to $V$ only once the larger groups have been satisfied. The satisfaction values under the two myopic rules jump heavily from one iteration to the next, as these rules tend to mainly represent one group of voters per iteration. 
On the other hand, the two dynamic rules keep the satisfaction of~$V$ relatively constant at around one fourth of the maximum possible satisfaction. 
These rules provide proportional representation in each single iteration, which is in line with the theoretical results in \Cref{sec:exp_guarantee}.

\section{Conclusion}

Motivated by the problem of how submitted questions in a live Q\&A session can be ranked in a more representative way, we have introduced dynamic ranking rules. We proposed two paradigms of dynamizing existing ranking rules: under the \textit{dynamic paradigm}, we target proportional representation of voter interests at each individual time step; under the \textit{myopic paradigm}, we try to make the set of implemented candidates as representative as possible. While the former approach lends more flexibility for the decision maker and guarantees a proportional exposure of candidates in each ranking, the latter approach is computationally more efficient and yields stronger selection guarantees.
Our experimental results illustrate the difference between the two approaches, and verify that both approaches lead to proportional results. 

The application of live Q\&A platforms gives rise to some interesting extensions of our model. In realistic scenarios, neither the electorate nor the set of candidates is static, as people enter or leave the audience and new questions come up continuously. Moreover, participants can change their approval preferences throughout the event. Our approach can take these dynamic aspects into account in a straighforward manner: After each implementation, we can apply our ranking rules to the current set of not-yet-implemented candidates and to the current approval preferences---the only necessary information from previous iterations is the sequence of implemented candidates. 

The dynamic ranking rules proposed in this paper are applicable to a wide variety of sequential selection procedures in which proportional representation is desired and, at the same time, some flexibility on the part of the decision maker is necessary (e.g., think of human-in-the-loop decision support systems for hiring or budgeting decisions). 
Other applications of dynamic ranking rules include committee election scenarios in which some part of the committee is fixed (e.g., due to external constraints) and the remaining seats need to be filled in such a way that the committee as a whole is representative.

\section*{Acknowledgments}

We would like to thank IJCAI and COMSOC reviewers for their insightful comments. 
This work was partially supported by the Deutsche Forschungsgemeinschaft under grant BR~4744/2-1. 
We thank 
Paula Blechschmidt,
Cristina Cornelio,
Benny Kimelfeld,
Phokion Kolaitis, and
Julia Stoyanovich
for helpful comments and discussions.

\bibliographystyle{plainnat}


\clearpage
\appendix

\section{Additional Details on \Cref{sec:dynamic_rules}}
\label{app:rules}
This section contains more details on Here we study the introduced dynamic ranking rules algorithmically. For each of the rules we provide a pseudocode formulation and give bounds on its running time. We start with the two rules based on sequential PAV and then consider the two Phragmén variants.

All four rules take as input an approval profile $A$ and a sequence of already implemented candidates $X$. We denote the voters given implicitly by the profile by $N(A)$ and the respective candidates by $C(A)$. The rules output a ranking of all candidates that have not been implemented yet.

In the pseudocode we make use of the following conventions. We set $n \coloneqq |N(A)|$ as the number of voters and $m \coloneqq |C(A)|$ as the number of candidates.
Further, we use set-syntax also for sequences of objects. For example, $(i \in \mathbb{Z} \mid \text{sorted non-decreasingly})$ denotes the (ordered) sequence of all positive natural numbers in non-decreasing order, i.e., the sequence $(1,2,3,\ldots)$.\footnote{If there are ties in the orderings, this notation does not uniquely define a single sequence. Since our results do not rely on a specific tie-breaking rule, we sidestep this ambiguity by assuming a given tie-breaking order.} 
Additionally, we abbreviate \emph{sorted non-decreasingly} to \emph{sorted $\nearrow$} and \emph{sorted non-increasingly} to \emph{sorted $\searrow$}. We are then able to denote $(1,2,3,\ldots)$ by $(i \in \mathbb{Z} \mid \text{sorted}\nearrow)$.

\subsection{Algorithmic Aspects of Dynamic and Myopic seqPAV}
We begin with the two rules based on sequential PAV. They both use the notion of marginal contribution to rank candidates. For a set $S\subseteq C$ of candidates define
\[ \tsc(S) = \sum_{i \in N} \sum_{j=1}^{|A_i\cap S|} \frac{1}{j}. \]
This score basically models a voter's satisfaction with $S$ using the concept of diminishing returns. A voter $i \in N$ contributes 1 to the score for the first candidate in $A_i \cap S$. For the second candidate in $A_i \cap S$ that voter contributes only 1/2 of additional score and so on.
Using this score the marginal contribution of a candidate $c \notin S$ can then be computed by
\[ \mc(c) = \tsc(S \cup \Set{c}) - \tsc(S). \]

Dynamic seqPAV mimics the non-dynamic variant closely by computing the marginal contribution of a candidate according to all already ranked and already implemented candidates. The rule ranks candidates greedily one by one. In every round it selects a candidate with maximal marginal contribution and appends that candidate to the end of the ranking.
\begin{algorithm}[h]
    \DontPrintSemicolon
    \SetKwInOut{Input}{Input}
    \SetKwInOut{Output}{Output}
    \newcommand\commfont[1]{\small\texttt{#1}}
    \SetCommentSty{commfont}
    \Input{approval profile $A$, sequence of implemented candidates $X$}
    \Output{ranking $r \in \mathcal{L}(C(A) \setminus X)$}
    $C \coloneqq C(A) \setminus X$ \tcp*{unranked candidates}
    $r \coloneqq ()$ \tcp*{ranking of candidates}
    \While{$C \neq \emptyset$}{
        \ForAll{$c \in C$}{
            $\mc_\mathrm{dyn}(c) = \tsc(X \cup r \cup c) - \tsc(X \cup r)$ \tcp*{compute marginal contribution}
        }
        choose $c \in \arg\max_{c \in C} \mc_\mathrm{dyn}(c)$ \tcp*{choose candidate with max.\ $\mc$}
        append $c$ to $r$ \tcp*{add that candidate to $r$}
        $C \coloneqq C \setminus \Set{c}$\;
    }
    \Return{$r$}
    \caption{Dynamic sequential PAV}\label{alg:dyn_seqPAV}
\end{algorithm}
\begin{proposition}
Given an approval profile $A$ and a sequence of implemented candidates $X$, dynamic seqPAV outputs a ranking of all not-yet-implemented candidates in time $\mathcal{O}(m^3n)$.
\end{proposition}
\begin{proof}
Termination of the algorithm is straightforward. 
Concerning running time, first note that $\mc_\mathrm{dyn}(c)$ can be computed in time $m \cdot n$. The loops starting in Lines 3 and 4 of \Cref{alg:dyn_seqPAV} each have length of at most $m$. Choosing a candidate with maximum marginal contribution (Line 6) can be done in an additional time of $m$. This however gets dominated by the running time of the loop in Line 4. Appending and deleting candidates in Lines 7 and 8 is possible in constant time.
Thus, the overall running time is in $\mathcal{O}(m^3n)$.
\end{proof}

The myopic version of this rule computes the marginal contribution of the candidates once up front---only with respect to the already implemented candidates---and then simply ranks all candidates according to this score.
\begin{algorithm}[h]
    \DontPrintSemicolon
    \SetKwInOut{Input}{Input}
    \SetKwInOut{Output}{Output}
    \newcommand\commfont[1]{\small\texttt{#1}}
    \SetCommentSty{commfont}

    \Input{approval profile $A$, sequence of implemented candidates $X$}
    \Output{ranking $r \in \mathcal{L}(C(A) \setminus X)$}
    $C \coloneqq C(A) \setminus X$ \tcp*{unranked candidates}
    \ForAll{$c \in C$}{
        $\mc_\mathrm{myopic}(c) = \tsc(X \cup c) - \tsc(X)$ \tcp*{compute marginal contribution}
    }
    $r \coloneqq (c \in C \mid \text{sorted$\searrow$ by } \mc_\mathrm{myopic}(c))$ \tcp*{\parbox[t]{16em}{\raggedright rank candidates by $\mc_\mathrm{myopic}$ in non-in\-creasing order}}

    \Return{$r$}
    \caption{Myopic sequential PAV}\label{alg:myopic_seqPAV}
\end{algorithm}
\begin{proposition}
Given an approval profile $A$ and a sequence of implemented candidates $X$, myopic seqPAV outputs a ranking of all not-yet-implemented candidates in time $\mathcal{O}(m^2n)$.
\end{proposition}
\begin{proof}
Termination of the algorithm is straightforward. 
Concerning running time, first note that $\mc_\mathrm{myopic}(c)$ can be computed in time $m \cdot n$. The loop starting in Line 2 of \Cref{alg:myopic_seqPAV} has length of at most $m$. Ranking (i.e., sorting) all candidates takes an additional time of $m \log(m)$ which is however dominated by the running time of the above loop.
Thus, the overall running time is in~$\mathcal{O}(m^2n)$.
\end{proof}

\subsection{Algorithmic Aspects of Dynamic and Myopic Phragmén}

\newcommand{\compdebts}{\texttt{compute\_debts}\xspace}
\newcommand{\comptime}{\texttt{compute\_buying\_time}\xspace}

The two Phragmén variants defined in \Cref{sec:dynamic_rules} are based on similar ideas. We assume that a candidate costs 1 credit and voters approving a candidate can buy this candidate with the credits they own. 
The first version, dynamic Phragmén, closely resembles the original, non-dynamic rule sequential Phragmén. In the original rule, voters start out with a balance of 0 and earn money continuously and at the same rate. At the moment a set of voters who approve the same candidate together have 1 credit, they spend this money immediately to buy this candidate. The newly bought candidate gets appended to the ranking. The main difference in the dynamic adaption is that voters do not start out with 0 credits in the beginning but might have a negative balance to start with. This enables us to model that certain voters are already satisfied through the implemented candidates in~$X$. Once the initial debts are assigned voters again earn money continuously. As soon as a set of voters together have a (positive) amount of 1 credit they can spend this on a commonly approved candidate. Voters with initial debts have to wait until they earned enough credits to cover their debts---and thus have a positive balance---until they can participate in buying candidates. It is not possible to go into debt to buy a candidate. For a voter $i \in N(A)$ we denote $i$'s credits with $\credit_i$.

In the following we present a pseudocode formulation of dynamic Phragmén as \Cref{alg:dyn_phrag}. \Cref{alg:dyn_phrag} uses two subroutines that we will describe in more detail afterwards.
First, to compute the initial debts for dynamic Phragmén (and myopic Phragmén, which we discuss later in this section), we use a subroutine called \compdebts (see \Cref{alg:compute_debts}). This algorithm takes an approval profile $A$ and a sequence of implemented candidates $X$ as input and outputs the amount of debt each voter $i \in N(A)$ receives to accommodate the costs of the already implemented candidates in $X$. 
Dynamic Phragmén then interprets these debts as negative credits, setting $\credit_i = -d_i$. Dynamic Phragmén constructs the output ranking iteratively. In order to find the next candidate to rank, the algorithm searches for the candidate that can be bought by (a subset of) its supporters at the earliest point in time.
To analyse the running time of dynamic Phragmén, it is easier to \textit{not} think of giving credits to voters continuously to find the next candidate. Rather, in each iteration of the ranking process, we calculate for each unranked candidate $c$ the minimal time we have to wait until the supporters of $c$ can buy $c$. Then we rank the candidate with the smallest such time next and let the corresponding supporters pay for $c$. The calculation of this minimal time is done by the subroutine \comptime. 
\begin{algorithm}[h]
    \DontPrintSemicolon
    \SetKwInOut{Input}{Input}
    \SetKwInOut{Output}{Output}
    \newcommand\commfont[1]{\small\texttt{#1}}
    \SetCommentSty{commfont}
    \Input{approval profile $A$, sequence of implemented candidates $X$}
    \Output{ranking $r \in \mathcal{L}(C(A) \setminus X)$}
    $C \coloneqq C(A) \setminus X$ \tcp*{unranked candidates}
    $N \coloneqq N(A)$ \tcp*{voters}
    $r \coloneqq ()$ \tcp*{ranking of candidates}
    $(d_i)_{i \in N} \coloneqq \mathrm{\compdebts}(A,X)$ \tcp*{compute initial debts}
    \ForAll{$i \in N$}{
        $\credit_i \coloneqq -d_i$ \tcp*{starting credits}
    }
    \While{$C \neq \emptyset$}{
        \ForAll{$c\in C$}{
            $(t_c, V_c) \coloneqq \mathrm{\comptime}(A,c,(\credit_i)_{i \in N})$\;
        }
        choose $c \in \arg\min_{c \in C} t_c$\tcp*{choose candidate with min. $t_c$}
        \ForAll{$i \in V_c$}{
            $\credit_i = 0$ \tcp*{voters in $V_c$ pay for candidate $c$}
        }
        append $c$ to $r$ \tcp*{add that candidate to $r$}
        $C \coloneqq C \setminus \Set{c}$\;
    }
    \Return{$r$}
    \caption{Dynamic Phragmén}\label{alg:dyn_phrag}
\end{algorithm}

\paragraph{Subroutine \compdebts.}
This algorithm computes the initial debts for both Phragmén variants (see \Cref{alg:compute_debts}). The input of the algorithm is an approval profile $A$ and a sequence of implemented candidates $X$ and it outputs for each voter $i \in N(A)$ a non-negative real number $d_i$ which is the amount of debt voter $i$ received to accommodate the costs of the already implemented candidates in $X$. Thus it holds $\sum_{i \in N(A)} d_i = |X|$.
The algorithm iterates over $X$ and for each candidate distributes the cost of 1 among all voters that approve of that candidate. More formally, let $x \in X$ be an already implemented candidate and let $N_x$ be the supporters of $x$. We divide the cost of 1 for $x$ among $N_x$ in a way that minimises the total debt across all voters in $N_x$. This is the same approach that is used in the load-balancing version of sequential Phragmén. New in our case is that the difference between debt (or load) of two voters assigned in previous steps of the algorithm can be arbitrarily high. 
(For example, if the first $k$ implemented candidates in $X$ are only supported by a single voter, this voter gets assigned a debt of $k$ before any other voter gets assigned any debt.)
This may lead to a situation where, in order to minimise the maximal total debt across voters in $N_x$, we distribute the cost only over a subset of $N_x$. Thus, for each $x \in X$ we first have to determine the correct subset of $N_x$ to distribute the debt to. We do this by ordering voters in $N_x$ non-decreasingly by $d_i$ and inspect any prefix of voters in this order. For every such prefix $N_x' \subseteq N_x$ we compute the debt voters in $N_x'$ have after distributing the additional debt of 1 by 
\[d_\mathrm{new} \coloneqq \frac{1 + \sum_{i \in N_x'}d_i}{|N_x'|}. \]
Let $v_{+1} \in N_x \setminus N_x'$ be that supporter of $x$ not in $N_x'$ with the next lowest debt $d_{v_{+1}}$. It is possible to lower $d_\mathrm{new}$ by adding $v_{+1}$ to $N_x'$ if and only if $d_{v_{+1}} < d_\mathrm{new}$.
The algorithm \compdebts does this computation for each $x \in X$ in the order given by the sequence $X$ itself and then outputs the debts computed this way.
\begin{algorithm}[h]
    \DontPrintSemicolon
    \SetKwInOut{Input}{Input}
    \SetKwInOut{Output}{Output}
    \newcommand\commfont[1]{\small\texttt{#1}}
    \SetCommentSty{commfont}
    \Input{approval profile $A$, sequence of implemented candidates $X$}
    \Output{$(d_i)_{i \in N(A)}$, debts for all $i \in N(A)$}
    $N \coloneqq N(A)$ \tcp*{voters}
    \ForAll{$i \in N$}{
        $d_i \coloneqq 0$ \tcp*{initial debt}
    }
    \For{$x \in X$}{
        $N_x \coloneqq (i \in N \mid x \in A_i, \text{sorted$\nearrow$ by } d_i)$ \tcp*{\parbox[t]{17em}{\raggedright supporters of $x$ sorted non- decreasingly by initial debt}}
        \For{$j \leq |N_x|$}{
            $N_x' \coloneqq (N_x)_{\leq j}$ \tcp*{first $j$ voters in $N_x$}
            $d_\mathrm{new} \coloneqq \frac{1 + \sum_{i \in N_x'}d_i}{j}$ \tcp*{distribute debt among first $j$ supporters}
            \If(\tcp*[f]{check next supporter's debt}){$j < |N_x|$ and $d_\mathrm{new} \leq d_{j+1}$}{
                \ForAll{$i \leq j$}{
                    $d_i = d_\mathrm{new}$ \tcp*{assign new debt}
                    break \tcp*{break inner for-loop}
                }
            }
            \If(\tcp*[f]{all supporters share debt}){$j = |N_x|$}{ 
                \ForAll{$i \in N_x$}{
                    $d_i = d_\mathrm{new}$ \tcp*{assign new debt}
                }
            }
        }
    }
    \Return{$(d_i)_{i \in N(A)}$}
    \caption{\compdebts}\label{alg:compute_debts}
\end{algorithm}
\begin{proposition}\label{prop:compute_debts}
Given an approval profile $A$ and a sequence of implemented candidates $X$, \Cref{alg:compute_debts} computes the debts for all voters in $N(A)$ according to $X$ in time $\mathcal{O}(mn^2)$.
\end{proposition}
\begin{proof}
Concerning termination of \Cref{alg:compute_debts}, consider a candidate $x \in X$ and the sequence of corresponding supporters $N_x \subseteq N(A)$. The sequence $N_x$ gets sorted non-increasingly in Line 5. In the loop starting in Line 6 each prefix of voters approving $x$ is considered. For each such prefix $N_x'$ the algorithm calculates the debt that each voter in $N_x'$ would get if they together had to pay 1 credit for candidate $x$ (Line 8). As described above, the algorithm then checks whether this debt can be decreased by adding the next candidate in $N_x \setminus N_x'$ to $N_x'$ (Line 9). If not, the debt is distributed among voters in $N_x'$. Otherwise, the loop continues. The if-statement in Line 13 guarantees that in the end, if no proper subset of candidates $N_x' \subset N_x$ was singled out, all supporters of $x$ share the debt.

Concerning running time, the loop starting in Line 4 has length of at most $m$. Sorting candidates into $N_x$ can be done in time $n \log(n)$. On the other hand, the loop in Line 6 has length at most $n$ and the calculation in Line 8 can be performed in $\mathcal{O}(n)$ operations. The assignment of debts in Lines 11 or 15 would contribute another $n$ operations, which are however dominated by the calculation in Line 8.
Since the running time of the loop starting in Line 6 dominates the running time of the operation in Line 5, the overall running time of the algorithm is in $\mathcal{O}(mn^2)$.
\end{proof}

\paragraph{Subroutine \comptime.}
As stated above, it is easier to not think of giving credits to voters continuously when analysing the running time of dynamic Phragmén. Instead in each iteration of the ranking process, we calculate for each unranked candidate $c$ the minimal time we have to wait until the supporters of $c$ can buy $c$. Then we rank the candidate with the smallest such time next and let the corresponding supporters pay for $c$. 
Note that not necessarily all supporters of $c$ have to pay, as it might be the case that some of them have to much debt and the rest of the supporters are able to raise 1 credit before those with a lot of debt can participate in the buying process.
To compute that minimal time and the corresponding set of supporters of a candidate $c$, we use the subroutine \comptime. This algorithm takes as input the approval profile $A$, a candidate $c \in C(A)$ and the current credit balance of all voters $(\credit_i)_{i \in N}$. The function outputs the minimal additional time until a subset of supporters of candidate $c$ will have accumulated a positive budget of 1 credit. Formally, the function finds the minimal $t_c \in [0,1]$ such that there is a set of voters $V \subseteq N_c$ with
\[ \sum_{i \in V} (\credit_i + t_c) \geq 1. \]
Because of the minimality of $t_c$ we have $\credit_i + t_c \geq 0$ for all $i \in V$.
Practically, $t_c$ and the corresponding set of supporters can be computed as follows. Set $V = \Set{i \in N_c \mid \credit_i \geq 0}$ and sort all remaining supporters $i \in N_c \setminus V$ by non-increasing budget (i.e., voters with less debt are sorted to the top). For each $k = 0,1,\ldots,|N_c \setminus V|$ consider the set $V_k$ containing voters in $V$ and the first k voters in $N_c\setminus V$. That means $V_0 = V$ and for $k = 1,2,3,\ldots$ the set $V_k$ will additionally contain the $k$ voters in $N_c \setminus V$ with the fewest debt.
For each of these $V_k$ compute
\[ t_c^k \coloneqq 
\begin{cases} 
    0, &\text{ if } \sum_{i \in V_k} \credit_i \geq 1 \\
    \frac{1 - \sum_{i \in V_k}\credit_i}{|V_k|}, &\text{ else.}
\end{cases} \]
The function \comptime then outputs the lowest $t_c^k$ and the corresponding set of voters $V_k$, breaking ties by smaller $k$.
We now consider the running time of this algorithm. Setting up the sets $V$ and sorting $N_c \setminus V$ can be done in $\mathcal{O}(n\log(n))$. This gets dominated by the running time of the loop over all $k \leq |N_c \setminus V|$. This loop has length at most $n$ and for each $k$ the calculation of $t_c^k$ is possible in $\mathcal{O}(n)$. Thus \comptime runs in time $\mathcal{O}(n^2)$.
We have therefore proven the following result.
\begin{proposition}\label{prop:compute_buying_time}
Given an approval profile $A$, a candidate $c \in C(A)$ and the credit balance of all voters $(\credit_i)_{i \in N}$, \comptime outputs $t_c$ and $V_c$ in time $\mathcal{O}(n^2)$, where $t_c$ is the minimal additional time until a subset of voters $V_c \subseteq N_c$ have a joint budget of 1 credit.
\end{proposition}
With this, we can prove the following result concerning the running time of dynamic Phragmén.
\begin{proposition}
Given an approval profile $A$ and a sequence of implemented candidates $X$, \Cref{alg:dyn_phrag} computes a ranking of all not-yet-implemented candidates in time $\mathcal{O}(m^2n^2)$.
\end{proposition}
\begin{proof}
Termination of the algorithm is straightforward.
Regarding running time we begin with the loops starting in Lines 7 and 8 of \Cref{alg:dyn_phrag}. Both loops have length at most $m$. By \Cref{prop:compute_buying_time} we know that computing the buying time in Line 9 can be done in $\mathcal{O}(n^2)$. This (in conjunction with the loop in Line 8) dominates the running time of choosing a candidate with minimum $t_c$ (Line~10) and the loop in Line 11. Further, also the running time of the subroutine \compdebts of $\mathcal{O}(mn^2)$ is dominated by this.
Thus, the overall running time is in $\mathcal{O}(m^2n^2)$.
\end{proof}

Myopic Phragmén also uses the debts from \compdebts, but in a greedy way. Here, we first iterate over all candidates that are to be ranked and for each $c \in C(A) \setminus X$ do the following. We append $c$ to the sequence $X$, obtaining the sequence $X_c$, and then compute the debts for all voters according to $X_c$. We then rank those candidates highest for which the maximal voter debt is lowest. This, in a way, checks which candidate is most suited to be implemented next. Formally, for each $c \in C(A) \setminus X$ we first sort the debts $(d_i^c){i \in N}$ of all voters according to $X_c$ non-increasingly. Then we rank the candidates $c \in C(A) \setminus X$ by comparing the sorted vectors $(d_i^c)_{i \in N}$ lexicographically, ranking lexicographically smaller candidates (i.e., candidates that induce a lower maximum debt) higher.
\begin{algorithm}[h]
    \DontPrintSemicolon
    \SetKwInOut{Input}{Input}
    \SetKwInOut{Output}{Output}
    \newcommand\commfont[1]{\small\texttt{#1}}
    \SetCommentSty{commfont}
    \Input{approval profile $A$, sequence of implemented candidates $X$}
    \Output{ranking $r \in \mathcal{L}(C(A) \setminus X)$}
    $C \coloneqq C(A) \setminus X$ \tcp*{unranked candidates}
    $N \coloneqq N(A)$ \tcp*{voters}
    \ForAll{$c \in C$}{
        $X_c \coloneqq X.\mathrm{append(c)}$ \tcp*{append $c$ to $X$}
        $(d^c_i)_{i \in N} \coloneqq \mathrm{\compdebts}(A,X_c)$ \tcp*{compute debts w.r.t.\ $X_c$}
        $(\hat{d}^c_i)_{i \in N} \coloneqq$ sorted$\searrow$ $(d^c_i)_{i \in N}$\tcp*{sort debts  non-increasingly}
    }
    $r \coloneqq (c \in C \mid \text{sorted$\nearrow$ lexicographically by } (\hat{d}^c_i)_{i \in N})$\tcp*{\parbox[t]{13.5em}{\raggedright rank candidates w.r.t.\ max.\ incurred debt}}
    \Return{$r$}
    \caption{Myopic Phragmén}\label{alg:myopic_phrag}
\end{algorithm}
\begin{proposition}
Given an approval profile $A$ and a sequence of implemented candidates $X$, myopic Phragmén outputs a ranking of all not-yet-implemented candidates $c \in C(A) \setminus X$ in time $\mathcal{O}(m^2n^2)$.
\end{proposition}
\begin{proof}
Termination of the algorithm is straightforward. 
Regarding running time we start with the loop in Line 3 of \Cref{alg:myopic_phrag} which has length at most $m$. Line 4 takes constant time and Line 5 can be performed in $\mathcal{O}(mn^2)$ operations, as shown in \Cref{prop:compute_debts}. Sorting the debts w.r.t.\ $X_c$ (Line 6) needs additional $n \log(n)$ operations, which is dominated by the computation of the debts.
Lastly, ranking all candidates lexicographically w.r.t.\ $(d_i^c)_{i \in N}$ can be done in $\mathcal{O}(n m\log(m))$ operations. This is however dominated by the loop starting in Line 3.
Thus the overall running time is in $\mathcal{O}(m^2n^2)$.
\end{proof}

We note that the (asymptotic) running time for computing myopic Phragmén can be improved based on the following observation. For two distinct candidates $c, c' \in C(A)$, the sequences of candidates $X_c$ and $X_{c'}$ only differ in the last entry. This is because $c$ and $c'$ both get appended to the same sequence $X$.
When computing the debts w.r.t.\ $X_c$ in Line 5 of \Cref{alg:myopic_phrag}, the computation of the debts for all candidates in $X \subseteq X_c$ is the same, independent of the choice of candidate $c$. Thus, it is possible to compute the debts for all voters according to $X$ first, before entering the loop in Line 3, by calling $\mathrm{\compdebts}(A,X)$ once, and reuse these debts for each candidate $c \in C$ in the loop. To compute the debts w.r.t.\ $X_c$ in Line 5, only one more iteration of the calculations of \compdebts is needed (i.e., Lines 5 to 15 in \Cref{alg:compute_debts}). This is possible in time $\mathcal{O}(n^2)$, bringing the overall running time of myopic Phragmén down into $\mathcal{O}(mn^2)$.

\newpage
\section{Additional Details on \Cref{sec:impl_mono}}\label{app:implmono}

Here we study notions of monotonicity in the context of dynamic ranking rules more in-depth.
\subsection{Both Phragmén variants fail monotonicity}
We first consider myopic Phragmén. Recall the adapted example from the proof of \Cref{thm:implmono} for myopic seqPAV.
\begin{align*}
    2 &\times \Set{a}, ~ 15 \times \Set{a,b}, ~ j + 6 \times \Set{b}, ~ 10 \times \Set{c}, \\
    10 &\times \Set{d}, ~ j + 6 \times \Set{a,c,d}, ~ j + 16 \times \Set{e}.
\end{align*}
We will now argue that myopic Phragmén also fails to satisfy $(h,\alpha)$-monotonicity on this example for all $j\in \mathbb{N}$. This rule computes for each candidate $c \in C\setminus X$ the debt of voters that is induced by buying the candidates in $X$ of the current iteration (in the order of implementation) and afterwards buying candidate $c$. All candidates then get ranked by comparing the so computed debts of the voters lexicographically.
In the first iteration $X = ()$ holds and thus myopic Phragmén is equivalent to AV. The ensuing ranking is $r^1 = (a,b,c,d,e)$, independent of $j$. Now assume that the DM implements candidate $x_1 = b$. In the second iteration thus each supporter of $b$ has a debt of $\frac{1}{21 + j}$, since there are $21 + j$ voters who approve $b$ and they all share the price of 1 credit for buying $b$ into the ranking. Note that in this example it is always favorable to balance the debt induced by a candidate equally among its supporters (this might not be the case if a candidate that is only supported by very few voters got implemented before). Thus for every candidate $c \in C\setminus X$ we can compute the debt each of its supporters would have if $c$ would be bought next by $s_c^{(2)} = \frac{1}{|N_c|}(1+\sum_{i\in N_c} d_i)$, where $d_i$ is voter $i$'s debt induced by $X$. 
Myopic Phragmén now ranks the candidates in non-increasing order of $s_c^{(2)}$ since this is the relevant part in comparing the debts of all voters by their maximum as described in the definition of myopic Phragmén. We can compute $s^{(2)}_a = \frac{36+j}{(23+j)(21+j)}$ and $s_c^{(2)} = s_d^{(2)} = s_e^{(2)} = \frac{1}{16+j}$. To prove that for all $j \in \mathbb{N}$ myopic Phragmén fails to satisfy group implementation monotonicity in this example we have to validate that $s_a^2 > s_e^2$ holds independent of~$j$, which can be done by basic calculus.

We can show the claim regarding dynamic Phragmén in a similar manner. For that we use the example as presented in the proof of \Cref{thm:implmono}. Here the computation gets far more technical as the debts that get compared during the ranking process of the candidates change in each step of the ranking and not just once per iteration as was the case for the myopic variant of the rule. Nevertheless, we obtain a system of inequalities that ensures that the candidates get ranked by dynamic Phragmén in a similar way as by dynamic sequential PAV (and thus dynamic Phragmén also violates the monotonicity axiom). We can then again check that these inequalities hold for all $j$.

\subsection{Weaker Version of Monotonicity}\label{app:weakmono}
Here we want to consider a weaker form of $(h,\alpha)$-monotinicity. Recall that in \Cref{sec:impl_mono} the group we considered for checking monotonicity did not approve of the implemented candidate but that there was always a (sizable) part of the electorate that approved of the implemented candidate and at least one of the candidates that the group we considered approved of. We will now show that we can not recover $(h,\alpha)$-monotonicity if we only consider election where there is no voter that approves of the implemented candidate and any of the candidates the group we consider approves of for the dynamic rules. On the other hand the myopic counterparts satisfy this weaker version of the axiom. Formally we define weak $(h,\alpha)$-monotonicity as follows.

\begin{definition}
For $h \geq 1$ and $\alpha \in (0,1]$, a dynamic ranking rule satisfies \emph{weak $(h,\alpha)$-monotonicity} if, for all profiles and all groups of voters $V \subseteq N$ of size $|V| \geq \alpha\cdot |N|$, the following holds. For every iteration $t$ where there is no $c \in \bigcup_{k \in V} A_k$ with $\{c,x_t\}\subseteq A_i$ for some $i \in N$
we have
\[ \avg_V(r^{t+1}_{\leq h}) \geq \avg_V(r^t_{\leq h}). \]
\end{definition}
We will now show that both dynamic rules fail to satisfy even this weaker axiom but then argue that the myopic rules satisfy it.

\paragraph{Dynamic sequential PAV.}
For dynamic seqPAV consider the following profile of 177 voters and 5 candidates.
\begin{align*}
    4 &\times \Set{a}, ~ 27 \times \Set{a,b}, ~ 27 \times \Set{b}, ~ 30 \times \Set{c}, \\
    9 &\times \Set{c,d}, 9 \times \Set{d}, ~ 36 \times \Set{a,d}, ~ 35 \times \Set{e}.
\end{align*}
The  rule outputs the ranking $r^1 = (a,b,c,e,d)$ in the first iteration. If we now assume that the DM implements candidate $x_1 = b$ again then dynamic sequential PAV outputs $r^2 = (d,a,e,c)$. Now consider the group of voters $V$ that consists of the 39 supporters of $c$ (i.e., the 30 voters that approve only of $c$ and the 9 voters that approve of $c$ and $d$). For $h =3$ we have $\avg_V(r^1_{\leq h}) = 1$ but $\avg_V(r^2_{\leq h}) = \frac{9}{39}$ which is a violation of the above axiom.
We can see that this again holds for larger $h$ by introducing clones of candidate $e$ and its 35 supporters. 

Similar to what we did in the proof of \Cref{thm:implmono} we can increase the relative size of $V$ without changing the rankings dynamic seqPAV outputs. For that consider the following adapted profile, where $j = 2 \cdot y$ for some $y \in \mathbb{N}$.
\begin{align*}
    4 &\times \Set{a}, ~ \left(2x + 27\right) \times \Set{a,b}, ~ 27 \times \Set{b}, ~ 30 \times \Set{c}, \\
    \left(x + 9\right) &\times \Set{c,d}, ~ 9 \times \Set{d}, ~ 36 \times \Set{a,d}, ~ \left(\frac{x}{2} + 35\right) \times \Set{e}.
\end{align*}
For $j \rightarrow \infty$ we see that $\frac{|V|}{|N|} \rightarrow \frac{2}{7}$. Combining this with the cloning of candidate $e$ and its $\frac{x}{2} + 35$ supporters we obtain an example where weak $(h,\alpha)$-monotonicity is violated for every $h \geq 3$ by a group of size nearly $\frac{2}{4+h}$.

\paragraph{Dynamic Phragmén.}
We use a similar example as above where we only add an additional clone of candidate $e$ and its 35 supporters.
\begin{align*}
    4 &\times \Set{a}, ~ 27 \times \Set{a,b}, ~ 27 \times \Set{b}, ~ 30 \times \Set{c}, \\
    9 &\times \Set{c,d}, ~ 9 \times \Set{d}, ~ 36 \times \Set{a,d}, \\
    35 &\times \Set{e_1}, ~ 35 \times \Set{e_2}.
\end{align*}
In the first iteration dynamic Phragmén outputs the ranking $r^1 = (a, c, b, e_1, e_2, d)$. If we again assume that the DM implements candidate $x_1 = b$ then it outputs $r^2 = (d, e_1, e_2, c, a)$. We again consider the group of voters $V$ that consists of the 39 supporters of $c$. For $h =3$ we have $\avg_V(r^1_{\leq h}) = 1$ but $\avg_V(r^2_{\leq h}) = \frac{9}{39}$ which is a violation of the above axiom.
This again holds for larger $h$ via an introduction of clones of candidate $e_1$ and its 35 supporters. 
To increase the relative size of the voter group $V$ consider the following adaption of the instance for an even natural number $j = 2\cdot y$ with $y \in \mathbb{N}$.
\begin{align*}
    4 &\times \Set{a}, ~ \left(2x + 27\right) \times \Set{a,b}, ~ 27 \times \Set{b}, ~ 30 \times \Set{c}, \\
    (x &+ 9) \times \Set{c,d}, ~ 9 \times \Set{d}, ~ 36 \times \Set{a,d}, \\
    \Bigl(\frac{x}{2} &+ 35\Bigr) \times \Set{e_1}, ~ \left(\frac{x}{2} + 35\right) \times \Set{e_2}.
\end{align*}
By checking the inequalities arising from this example in the same manner as described earlier for the stronger axiom we can verify that this example works out for all $j\rightarrow \infty$. Again combining this with the idea of cloning candidate $e_1$ and its $\frac{x}{2} + 35$ supporters we obtain an example where weak $(h,\alpha)$-monotonicity is violated by dynamic Phragmén for every $h \geq 3$ by a group of size nearly $\frac{2}{5+h}$.

\paragraph{Myopic Rules.}
Note that the two myopic rules satisfy weak monotonicity. To see this, consider a group of voters $V$ and any candidate $c$ that a voter in $V$ approves. If $c$ is not supported by any voter (not necessarily in $V$) that also supports the candidate that gets implemented next then the voting power or the debt (depending on the rule we are interested in) of $c$'s supporters  does not change from this iteration to the next. Since the voting power (or debt) of the supporters of other candidates can only decrease (or increase, respectively), $c$'s position in the ranking can not drop when going from one iteration to the next.

\section{Additional Details on \Cref{sec:exp_guarantee}}\label{app:exposure}
We first introduce the \emph{proportionality degree} as defined by \citet{Skow18a}.
\begin{definition}
Let a depth restriction $h \leq |C|$ and a profile $A$ be given and let $g: \mathbb{N} \times \mathbb{N} \rightarrow \mathbb{R}$.
We say set of voters $V\subseteq N$ is $\ell$-large w.r.t.\ $h$ if $|V| \geq \ell \cdot \frac{n}{h}$. A dynamic ranking rule $\R$ satisfies $h$-proportionality degree of $g$ if for all $\ell$-large sets of voters $V$ and all iterations $t\in \mathbb{N}$ the ranking $r^{t+1} = \R(A,X^t)$ satisfies
\[ \lambda^t(V) \geq g(\ell, h) \Rightarrow \avg_V(r_{\leq h}) \geq g(\ell, h). \]
Let $\mathcal{G}_h$ be the set of all such $h$-proportionality degrees of $\R$ then we say $\R$ satisfies proportionality degree of 
\[ d_\R(\ell) \coloneqq \min_h \sup_{g\in \mathcal{G}_h} g(\ell,h). \]

This means that the proportionality degree of $\R$ is the best guarantee on the above objective that holds for all depth restrictions $h \leq |C|$.
\end{definition}
Note that as was the case with the $\kappa$ functions used for group representation that did not only depend on $\alpha$ and $\lambda$, but also on the set $V$ and on the sequence $X$ of previously implemented candidates, we will simplify notation for the proportionality degree in a similar manner.
While $\kappa$-group representation as defined in the main text describes how far down a ranking a group of voters $V$ has to go to be guaranteed a certain happiness, the proportionality degree lower bounds the happiness of $V$ given a certain depth restriction (or more precisely the ensuing $\ell$-value of the voter group).

\subsection{Proportionality of Rankings for Dynamic Phragmén}
In this section we will prove bounds on the two measures of proportionality for dynamic Phragmén. We start with the proportionality degree. As before, let $d_i$ be the debt agent $i$ starts with in dynamic Phragmén and let $d_\text{avg} = \frac{1}{|V|} \sum_{i \in V} d_i$ be the average starting debt of agents in $V$.
\begin{theorem}\label{thm:propdeg_dphrag}
Given integer $h$, dynamic Phragmén satisfies proportionality degree of
\[d(\ell) \geq \frac{\ell - 1}{2} - \frac{m}{2} - \frac{s\cdot |V|}{4},\]
where $m = |\bigcup_{i\in V} A_i \cap X|$ and $s = \sum_{i\in V} (d_i - d_\text{avg})^2$.
\end{theorem}
If $X = ()$ then $m=s=0$ and we have the proportionality degree that was also proved in \citep{Skow18a} for the non-dynamic setting.
We prove the results by a similar potential function approach as provided for the respective non-dynamic result given by \citet{Skow18a} while taking into account the added complexity of the dynamic setting.
\begin{proof}
Let the depth restriction $h \in \mathbb{N}$ and an iteration $t \in \mathbb{N}$ be given. We set $g(\ell, h) = \frac{\ell - 1}{2} - \frac{m}{2} - \frac{s\cdot |V|}{4}$. To make the prove more consistent with \citep{Skow18a}, assume w.l.o.g.\ that a candidate costs $n$ instead of 1 credit.
Towards a contradiction we assume that there is a set of candidates $r^t_{\leq h}$ that is ranked in the first $h$ positions of the ranking put out by dynamic Phragmén at iteration $t$ and a group of voters $V$ such that $\lambda^t \geq g(\ell, h)$ but $\avg_V(r^t_{\leq h}) < g(\ell, h)$.
We will now investigate the ranking process of dynamic Phragmén more closely. In order to do that we imagine the process of buying candidates and waiting for credits for the voters as a time-dependent process. (Note that this process reflects the ranking of $r^t$ for a fixed iteration $t \in \mathbb{N}$, we denote the time elapsing during this process by $\theta$.) For each point in time $\theta > 0$ of the ranking process at iteration $t$ we define $p_i(\theta)$ to be the amount of credits voter $i \in V$ possesses at that moment and $p_\text{avg}(\theta) = \frac{1}{|V|} \sum_{i \in V} p_i(\theta)$. 
With that we can define the potential function 
\[ \phi(\theta) = \sum_{i \in V} (p_i(\theta) - p_\text{avg}(\theta))^2.\]
In contrast to the proof given by \citet{Skow18a}, the voters in our setting might start with negative money, that is the debt assigned to them by the dynamic Phragmén rule because of implemented candidates in $X$. This means that at time $\theta=0$ we might have $p_i(\theta) < 0$ for some voters and thus the potential function at time $\theta=0$ is not 0 but might be some positive real number $\bar{s} = \phi(0) = \sum_{i\in V} (p_i(0) - p_\text{avg}(0))^2 = n^2 \cdot s$. (Here, $\bar{s}$ corresponds to $s = \sum_{i \in V}(d_i - d_\text{avg})^2$ scaled according to the new costs of $n$ credits per candidate.)

The proof now proceeds as the proof by \citet{Skow18a} for Theorem 2. After $h$ time units in the ranking process at most $h\cdot n$ credits can be amassed by all voters and the procedure can not stop before that point in time. Until then, the voters in $V$ have at most $|V| \cdot h$ credits.
Now fix a point in time $\theta$ at which a candidate supported by some voters in $V$ is bought and a voter $j \in V$ that pays for this candidate. That means that the average $p_\text{avg}(\theta)$ decreases by $\frac{p_j(\theta)}{|V|}$.
The calculation of $\Delta_\phi$ can be done in the same way as in \citep{Skow18a}, as it does not depend on the starting value of $\phi$. That gives us for a point in time $\theta$ in which a candidate is bought the following.
\begin{align*}
    \Delta_\phi &= \sum_{i \in V, j\neq i} \left( p_i(\theta) - \left(p_\text{avg}(\theta) \frac{p_j(\theta)}{|V|}\right)\right)^2 \\
    &\phantom{M} + \left( 0 - \left( p_\text{avg}(\theta) \frac{p_j(\theta)}{|V|}\right)\right)^2 - \sum_{i\in V} (p_i(\theta) - p_\text{avg}(\theta))^2 \\
    &= \sum_{i \in V} \left( p_i(\theta) - \left( p_\text{avg} - \frac{p_j(\theta)}{|V|} \right)\right)^2 \\
    &\phantom{M} + \left( p_\text{avg} - \frac{p_j(\theta)}{|V|} \right)^2 - \left( p_j(\theta) - \left( p_\text{avg} - \frac{p_j(\theta)}{|V|} \right)\right)^2 \\
    &\phantom{M} - \sum_{i\in V} (p_i(\theta) - p_\text{avg}(\theta))^2 \\
    &= \sum_{i \in V} \frac{p_j(\theta)}{|V|} \left( 2\cdot p_i(\theta) - 2\cdot p_\text{avg}(\theta) + \frac{p_j(\theta)}{|V|} \right)^2 \\
    &\phantom{M} - p_j(\theta)^2 + 2 \cdot p_j(\theta) \cdot \left( p_\text{avg}(\theta) - \frac{p_j(\theta)}{|V|} \right),
\end{align*}
where in the last step we used the binomial formulas. By definition we have 
\[ \sum_{i \in V} (2 \cdot p_i(\theta) - 2\cdot p_\text{avg}(\theta)) = 0 \]
which lets us conclude
\begin{align*}
    \Delta_\phi &= \frac{p_j(\theta)^2}{|V|} + p_j(\theta) \cdot \left(2\cdot p_\text{avg}(\theta)  - \frac{p_j(\theta)}{|V|} - p_j(\theta) \right) \\
    &= p_j(\theta) \cdot \left( 2 \cdot p_\text{avg}(\theta)  - \frac{p_j(\theta)(|V| + 1)}{|V|}\right). 
\end{align*}

Accordingly with \citet{Skow18a}, we can observe that at each time $\theta$ we have $p_\text{avg}(\theta) \leq \frac{n}{|V|} \leq \frac{h}{\ell}$. This is because otherwise voters in $V$ would have more than $n$ units of credits left and would have been able to buy a candidate they approve of at an earlier time of the ranking process. Using that fact and setting $y_{\theta,j} = p_j(\theta) - \frac{2 |V|}{|V| +1} \cdot \frac{h}{\ell}$ we obtain
\begin{align*}
    \Delta_\phi &\leq p_j(\theta) \cdot \left( \frac{2h}{\ell}  - \frac{y_{\theta,j}(|V| + 1)}{|V|} - \frac{2h}{\ell}\right) \\
    &= - y_{\theta,j}^2 \frac{|V|+1}{|V|} - 2y_{\theta,j} \frac{h}{\ell} \\
    &\leq - 2y_{\theta,j} \frac{h}{\ell}.  
\end{align*}
Following \citet{Skow18a} again, if $x_{\theta,j} > 0$, then $\phi$ decreases by at least $2|y_{\theta,j}| \cdot \frac{h}{\ell}$ and if $y_{\theta,j} \leq 0$, then $\phi$ increases by at most $2|y_{\theta,j}| \cdot \frac{h}{\ell}$. Since the potential value is always non-negative and starts at $s$, the net-change~$\sum \Delta_\phi$ has to be greater or equal to $-\bar{s}$, i.e.,
\[
-\bar{s} \leq \sum \Delta_\phi \leq \sum_{(\theta,j)\in \mathbb{R}\times V: j \text{ pays at time } \theta} -2 y_{\theta,j} \cdot \frac{h}{\ell}.
\]
Let 
\begin{align*}
    z &\coloneqq |\{(\theta,j)\in \mathbb{N}\times V: j \text{ pays at time } \theta\}| \\
    &\leq \sum_{i\in V} |r_{\leq h} \cap A_i|
\end{align*}
be the number of single payments the voters in $V$ issued during the whole procedure (which is less than or equal to the total satisfaction of the group $V$).
Rearranging the terms of the bound on $\bar{s}$ above and plugging in the definitions of $y_{\theta,j}$ and $z$ yields 
\begin{align*}
    \frac{\bar{s}}{2}\cdot \frac{\ell}{h} &\geq \sum_{(\theta,j)} p_j(\theta) - \frac{2|V|}{|V| + 1}\cdot \frac{h}{\ell} \\ 
    &= \sum_{(\theta,j)} p_j(\theta) - z\cdot \frac{2 |V|}{|V| + 1} \cdot \frac{h}{\ell},
\end{align*}
where $\sum_{(\theta,j)} p_j(\theta)$ is the total amount of credits voters in $V$ spend for candidates they approved. Our goal now is to lower bound $z$ and with that to lower bound the average satisfaction of voters in $V$.

For this, first note that we know that $\sum_{(\theta,j)} p_j(\theta) \geq |V| \cdot (h + |X|) - n - \sum_{i \in V} d_i$. Plugging this in the above equation and rearranging terms we obtain
\begin{align*}
    z &\geq \frac{|V| + 1}{2 |V|} \cdot \frac{\ell}{h} \cdot \left(|V|(h+|X|) -n -\frac{\bar{s}\cdot \ell}{2h} - \sum_{i\in V} d_i \right) \\
    &\geq \frac{1}{2} \cdot \frac{\ell}{h} \cdot \left(|V|(h+|X|) -n -\frac{\bar{s}\cdot \ell}{2h} - \sum_{i\in V} d_i \right) \\
    &\geq \frac{1}{2} \left( |V|(\ell -1) + \frac{|V|\cdot |X| \cdot \ell}{h} - \frac{\bar{s} \cdot \ell^2}{2h^2} - \frac{\ell}{h} \cdot \sum_{i\in V} d_i \right),
\end{align*}
where in the last step we used the fact that $n\cdot \frac{\ell}{h} \leq |V|$. We can now use this to get the desired lower bound on the average satisfaction of the voter group $V$.
\begin{align*}
    \frac{1}{|V|} &\sum_{i\in V} |r_{\leq h} \cap A_i| \geq \frac{1}{|V|} \cdot y\\
    &\geq \frac{\ell -1}{2} + \frac{|X|\cdot \ell}{2h} - \frac{1}{2n}\cdot \sum_{i\in V} d_i - \frac{\bar{s} \cdot \ell}{4n\cdot h} \\
    &\geq \frac{\ell -1}{2} - \frac{m}{2} - \frac{\bar{s}\cdot |V|}{4n^2} \\
    &= \frac{\ell -1}{2} - \frac{m}{2} - \frac{s\cdot |V|}{4},
\end{align*}
where we again used $n\cdot \frac{\ell}{h} \leq |V|$ and the fact that $m = |\bigcup_{i\in V} A_i \cap X| \geq \frac{1}{n} \sum_{i \in V} d_i$.
This contradicts our assumption that $\avg_V(r^t_{\leq h}) < g(\ell, h)$ and concludes the proof.
\end{proof}

This result is independent of the iteration $t$ which makes it rather strong. On the other hand the definition of proportionality degree (and thus this result) rely on a fixed $h$ to determine $\ell$-large groups of voters. But it is possible to translate the proportionality degree defined by \citet{Skow18a} into $\kappa$-group representation as defined by \citet{SLB+17a}. \citet{Skow18a} mentions this connection but to the best of our knowledge this is the first explicit translation from one of the proportionality measures to the other.
\begin{lemma}\label{lem:PDtoGR}
Let $\R$ be a (dynamic) ranking rule which satisfies proportionality degree of $d_\R(\ell)$ for all $\ell \in \mathbb{Q}$. Then $\R$ satisfies $\kappa$-group representation for 
\[ \kappa (\alpha, \lambda) = \left\lceil \frac{d_\R^{-1} (\lambda)}{\alpha} \right\rceil . \]
\end{lemma}
Note that normally the proportionality degree is defined for $\ell \in \mathbb{N}$. For technical reasons we need the function to hold for all rational $\ell$ which is however covered by our proof of \Cref{thm:propdeg_dphrag}.
\begin{proof}
Given a group of voters $V$ with proportion $\alpha = \frac{|V|}{|N|}$ and adapted cohesiveness $\lambda^t = |\bigcap_{i \in V} A_i \setminus X |$. Let $h = \left\lceil 1/\alpha \cdot d_\R^{-1}(\lambda^t)\right\rceil$ then by construction $V$ is $d_\R^{-1}(\lambda^t)$-large w.r.t. $h$ since
\[ |V| = \alpha \cdot |N| = \frac{|N|\cdot d^{-1}_\R(\lambda^t)}{(1/\alpha)\cdot d^{-1}_\R(\lambda^t)} \geq \frac{|N|}{h} \cdot d^{-1}_\R(\lambda^t). \]
Thus we can apply the proportionality degree with $\ell = d_\R^{-1}(\lambda^t)$ and obtain an average satisfaction for $V$ of
\[ \avg_V(r_{\leq h}) \geq d_\R(d_\R^{-1}(\lambda^t)) = \lambda^t. \]
Thus $\R$ satisfies $\kappa$-group representation for $\kappa(\alpha,\lambda^t) = h = \left\lceil1/\alpha \cdot d_\R^{-1}(\lambda^t) \right\rceil$.
\end{proof}
Plugging the proportionality degree from \Cref{thm:propdeg_dphrag} into this lemma we obtain the desired bound on the group representation as mentioned in \Cref{thm:gr_dphrag}.

\subsection{Proportionality of Rankings for Dynamic seqPAV}
We now consider dynamic seqPAV. We first prove the before mentioned bound on the group representation and afterwards consider the proportionality degree.
The proof of \Cref{thm:gr_dseqpav} follows the proof presented by \citet{SLB+17a} for the according result in the non-dynamic setting.
\begin{proof}[Proof of \Cref{thm:gr_dseqpav}]
Fix some $\alpha \in (0,1]$, $\lambda \in \mathbb{N}$ and profile such that in some iteration $t\in \mathbb{N}$ there exists a group of voters $V\subseteq N$ with $|V| \geq \alpha \cdot n$ and $\lambda^t(V) \geq \lambda$. Let $r$ be the ranking dynamic seqPAV outputs in iteration $t$ given the already implemented candidates $X$. Set $h \coloneqq \left\lceil \frac{2(\lambda + \avg_V(X) + 1)^2}{\alpha^2} \right\rceil$. Towards a contradiction assume that $\avg_V(r_{\leq h}) < \lambda$.
Now set 
\begin{align*}
z &\coloneqq |V| \cdot (\avg_V(r_{\leq h}) + \avg_V(X)) \\
&< |V| \cdot (\lambda + \avg_V(X)).
\end{align*}
In every step $k \in [h]$ of the ranking procedure of dynamic seqPAV (in iteration $t$) there exists at least one candidate $c \in \bigcap_{i \in V} A_i \setminus (X \cup r_{\leq h})$, i.e., a candidate that is neither ranked nor implemented but approved by all voters in $V$.
We now consider a step $k \in [h]$ of the ranking process of dynamic seqPAV. Let 
\[ a_i(k) \coloneqq |A_i \cap r_{\leq h}| + |A_i \cap X| ~~\text{and}~~ T(k) \coloneqq \sum_{i \in N} \frac{1}{a_i(k) + 1}. \]
Then $a_i(k) \leq a_i(k+1)$ and $T(k) \geq T(k+1)$ for all $k\in [h]$ with $n \geq T(1) \geq T(2) \geq ... \geq T(h+1) \geq 0$.
Further, for each $k \in [h]$ it holds 
\begin{align*}
    \sum_{i \in V} a_i(k) &= \sum_{i \in V}|A_i \cap r_{\leq k}| + |A_i \cap X| \\
    &\leq \sum_{i \in V}|A_i \cap r_{\leq h}| + |A_i \cap X| = z
\end{align*}
and thus it holds that
\begin{align*}
    \frac{z+|V|}{|V|} &\geq \frac{\sum_{i\in V} a_i(k) + |V|}{|V|} \\
    &= \frac{\sum_{i\in V} (a_i(k) +1)}{|V|} \\
    &\geq \frac{|V|}{\sum_{i \in V} \frac{1}{a_i(k) + 1}},
\end{align*}
where in the last step we used the arithmetic mean-harmonic mean inequality. Taking the inverse of this yields
\begin{align*}
    \sum_{i\in V} \frac{1}{a_i(k) + 1} &\geq \frac{|V|^2}{|V| + z} \\
    &> \frac{|V|^2}{|V| + |V| \cdot(\lambda+\avg_V(X)} \\
    &= \frac{|V|}{\lambda + \avg_V(X) + 1}.
\end{align*}
Now let $V'$ be the group of voters supporting candidate $c'$ that got ranked in round $k$ instead of $c$. Since dynamic seqPAV favored $c'$ over $c$ we know that
\begin{equation}\label{eq:gr_dseqpav_1}
   \sum_{i \in V'} \frac{1}{a_i(k) + 1} \geq \sum_{i \in V} \frac{1}{a_i(k) + 1} > \frac{|V|}{\lambda + \avg_V(X) + 1}. 
\end{equation}
We can now obtain 
\begin{align*}
    T(k) - T(k+1) &= \sum_{i \in V'} \left(\frac{1}{a_i(k) + 1} - \frac{1}{a_i(k) + 2} \right) \\
    &= \sum_{i \in V'} \left(\frac{1}{(a_i(k) + 1)(a_i(k) + 2)}\right) \\
    &\geq \sum_{i \in V'} \left(\frac{1}{2(a_i(k) + 1)^2}\right).
\end{align*}
Using the Cauchy-Schwarz inequality here we can bound this in the following way.
\begin{align*}
    \sum_{i \in V'} \left(\frac{1}{2(a_i(k) + 1)^2}\right) &\geq \frac{1}{2 |V'|} \left( \sum_{i \in V'} \frac{1}{a_i(k) + 1} \right)^2 \\
    &> \frac{1}{2 n} \left(\frac{|V|}{\lambda + \avg_V(X) + 1}\right)^2 \\
    &\geq \frac{\alpha^2\cdot n}{2(\lambda + \avg_V(X) + 1)^2},
\end{align*}
where in the second step we used \Cref{eq:gr_dseqpav_1}.
This holds for each $k \in [h+1]$ and thus we have
\begin{align*}
    T(1) - T(h+1) &= \sum_{k \in [h]} T(k) - T(k+1) \\
    &> h\cdot \frac{\alpha^2\cdot n}{2(\lambda + \avg_V(X) + 1)^2} \geq n.
\end{align*}
This yields $T(k+1) < T(1) - n \leq 0$ which is a contradiction.
\end{proof}
Note that it is not immediately clear how to convert a bound on the group representation into a result on the proportionality degree (i.e. whether the inverse version of \Cref{lem:PDtoGR} is possible). Still, using the same approach as in the proof above we are able to also prove a bound on the proportionality degree of dynamic seqPAV.

\begin{theorem}\label{thm:propdeg_dseqPAV}
Dynamic seqPAV satisfies an $h$-proportionality degree of 
\[ g(\ell,h) = \ell \cdot \sqrt{\frac{1}{2h}} - \avg_V(X) -1. \]
\end{theorem}
\begin{proof}
Fix some $h \leq |C|$ and let $r$ be the ranking dynamic seqPAV outputs for a given profile and set of already implemented candidates $X$.
Further let $V$ be a group of voters and $z \coloneqq \ell \cdot \sqrt{\frac{1}{2h}} - \avg_V(X) -1$. Towards a contradiction assume that $V$ is $\ell$-large w.r.t.\ $h$ and has cohesiveness $|\bigcap_{i \in V} A_i \setminus X| \geq z$ but the average satisfaction of voters in $V$ is only $\avg_V(r_{\leq h}) < z$. Thus there exists at least one candidate $c \in \bigcap_{i \in V} A_i \setminus (X \cup r_{\leq h})$, i.e., a candidate that is neither ranked nor implemented but approved by all voters in $V$.
We now consider the steps $k \in [h]$ of the ranking process of dynamic seqPAV. We again use 
\[ a_i(k) \coloneqq |A_i \cap r_{\leq h}| + |A_i \cap X| \qquad \text{and} \qquad T(k) \coloneqq \sum_{i \in N} \frac{1}{a_i(k) + 1}. \]
Then $a_i(k) \leq a_i(k+1)$ and $T(k) \geq T(k+1)$ for all $k\in [h]$ with $n \geq T(1) \geq T(2) \geq ... \geq T(h+1) \geq 0$.
Further, for each $k \in [h + 1]$ it holds using the arithmetic mean-harmonic mean inequality 
\[ \sum_{i\in V} \frac{1}{a_i(k) + 1} \geq |V|^2 \frac{1}{\sum_{i\in V}(a_i(k) + 1)} \]
and 
\begin{align*}
    \frac{1}{|V|^2} &\sum_{i \in V} a_i(k) + 1 \\
    &= \frac{1}{|V|^2} \sum_{i \in V} \left( |A_i \cap r_{\leq h}| + |A_i \cap X| + 1 \right) \\
    &= \frac{1}{|V|^2} \sum_{i \in V} (|A_i \cap X| + 1) + \frac{1}{|V|} \avg_V(r_{\leq h})\\
    &< \frac{1}{|V|^2} \sum_{i \in V} (|A_i \cap X| + 1) \\
    & \phantom{VV} + \frac{1}{|V|} \left( \ell \cdot \sqrt{\frac{1}{2h}} - \frac{1}{|V|} \cdot \sum_{i \in V} (|A_i \cap X| + 1) \right)\\
    &= \frac{1}{|V|} \cdot \frac{\ell}{h} \cdot \sqrt{\frac{h}{2}} = \frac{1}{n} \cdot \sqrt{\frac{h}{2}}.
\end{align*}
Plugging the second inequality into the first we obtain
\[ \sum_{i\in V} \frac{1}{a_i(k) + 1} > n \cdot \sqrt{\frac{2}{h}}. \]
Again, let $V'$ be the group of voters supporting candidate $c'$ that got ranked in round $k$ instead of $c$. Since dynamic seqPAV favored $c'$ over $c$ we know that
\begin{equation}\label{eq:propdeg_dseqpav_1}
   \sum_{i \in V'} \frac{1}{a_i(k) + 1} \geq \sum_{i \in V} \frac{1}{a_i(k) + 1} > n \cdot \sqrt{\frac{2}{h}}. 
\end{equation}
As in the previous proof we obtain using the Cauchy-Schwarz inequality
\begin{align*}
    T(k) - T(k+1) &\geq \frac{1}{2 |V'|} \left( \sum_{i \in V'} \frac{1}{a_i(k) + 1} \right)^2 \\
    &> \frac{1}{2 n} \cdot n \cdot \sqrt{\frac{2}{h}} = \frac{n}{h},
\end{align*}
where in the second step we used \Cref{eq:propdeg_dseqpav_1}.
This holds for each $k \in [h+1]$ and thus we have
\[ T(1) - T(h+1) = \sum_{k \in [h]} T(k) - T(k+1) > h \cdot \frac{n}{h} = n. \]
This yields $T(k+1) < T(1) - n \leq 0$ which is a contradiction.
\end{proof}
Note that this is the first known such bound in closed form for a seqPAV-variant. There are explicit bounds for small $h$ for the non-dynamic version provided by \citet{Skow18a}. The added generality of our closed form comes at the price of less accuracy when compared to those bounds. While \citet{Skow18a} shows that for $h=20$ the proportionality degree of (non-dynamic) seqPAV is greater or equal to $0.7503 \cdot \ell$ our bound only gives $0.1581\cdot \ell -1$ (when considering the non-dynamic case where $X=()$).

\subsection{Proportionality of Rankings for Myopic Rules}
We first provide a proof of \Cref{thm:gr_lazy} which uses the same counterexample as is provided by \citet{SLB+17a} for their Theorem~2.
\begin{proof}[Proof of \Cref{thm:gr_lazy}]
The first negative result follows simply by noting that both myopic rules are equal to AV if $X = ()$ and referring to the negative result for AV provided by \citet{SLB+17a}. We prove the second negative result by means of a counterexample which is again an adapted version of the one given by \citet{SLB+17a} for AV. Assume $\lambda^t, m$ and a function $\kappa(\alpha,\lambda)$ are given, set $\alpha \leq \frac{m+1}{m+2}$ and $h = \kappa(\alpha,\lambda^t)$.
Let $A = A_X \dotcup A_V \dotcup A_G$ with $|A_X| = m$ and $|A_V|=|A_G|=h$ be a set of candidates and $N = V \dotcup G$ an electorate composed of the disjoint union of voter groups $V$ and $G$ with $|V| < \alpha |N|$. Let the profile be such that all voters in $V$ approve of all candidates in $A_X$ and $A_V$ and all voters in $G$ approve of all $A_G$ and set $X = A_X$. Then both myopic rules rank all $a \in A_G$ higher than each candidate in $A_V$ and thus $\avg_V(r_{\leq h}) = 0$.
\end{proof}
Since AV and both myopic rules do not allow for any bound on $\kappa$-group representation for groups of voters that are not already a majority of the electorate it follows directly from \Cref{lem:PDtoGR} that they do not allow any bound on the proportionality degree for those groups either.
\begin{corollary}
AV does not satisfy any bound on the proportionality degree for groups of size $\alpha \leq \frac{1}{2}$.
Myopic seqPAV and myopic Phragmén do not satisfy any bound on the proportionality degree for groups of size $\alpha \leq \frac{m+1}{m+2}$, where $m \coloneqq |\bigcup_{i \in V} A_i \cap X|$.
\end{corollary}

\section{Additional Details on \Cref{sec:experiments}}\label{app:experiments}
We describe the models and experiments introduced in \Cref{sec:experiments} in more detail here. We start with the two models used to generate the approval profiles and then describe the selection procedure of the DM. 
Throughout our experiments we always used instances with 60 voters and 20 candidates.

\subsection{Generating Approval Profiles}
We first describe how we generate the approval profiles randomly. Since we want to study satisfaction of groups of voters with roughly similar approval preferences we somehow want to generate profiles with such voter groups. We used two different approaches.

\paragraph{Blurred Parties.}
Here we assign each of the 60 voters to one of two parties. The size of the parties will vary over the experiments and we will concentrate on the satisfaction of one of the parties. Additionally, we will also associate half of the candidates to one of the parties and the other half to the other, such that each party has 10 candidates associated with them. 
Here, a party is a group of voters that have the same probability of approving a candidate.
Now, for a voter we go over all candidates and say that the voter approves that candidate with probability 0.95, if it is a candidate of the voters party, and with probability 0.05, otherwise. This process is done independently for each voter and for each candidate. Thus, each voter in expectation approves of 95\% of the candidates in their party.

\paragraph{Spatial.}
This is an adaption of the 4-Gaussian model that was described by \citet{EFL+17a} for the setting of linear preferences. We again group voters and candidates into parties, using 3 parties this time. Thus two parties get associated 7 candidates and one party gets associated 6 candidates. The number of voters in the parties again varies but we are always concerned with the satisfaction of the first party, $V\subseteq N$, of size $|V|$ and set the sizes of the other two parties to $\left\lceil60-\frac{|V|}{2}\right\rceil$ and $\left\lfloor60-\frac{|V|}{2}\right\rfloor$, respectively. 
In this model we now take a spatial approach using the Euclidean plane. Each of the three parties gets assigned a point---their \emph{center}---that lies on the unit circle. To make the three points equidistant from each other we place them at 0, 120 and 240 degrees. Now the voters and candidates from each party get sampled as points on the Euclidean plane according to a 2-dimensional Gaussian (i.e., normal distribution) with standard deviation 0.4 around their party-center. We say a voter approves of a candidate if that candidate is at Euclidean distance at most 0.8. 
Note that if we assume that a voter gets assigned their party's center as point in the Euclidean plane, then by construction of the Gaussian distribution with standard distribution 0.4, this voter approves of roughly 95\% of the voters of their party in expectation.

\subsection{Selecting Candidates}
To decide which candidates the DM selects for implementation we also use a probabilistic approach. Since the ranking that is provided to the DM should somehow factor into the decision which candidate to select we try to mimic a realistic behavior. To this end we use an approximation of Google's so called \emph{click-through rates (CTRs)}. These rates describe how likely it is for a user to click on the first, second, and so on entry in a Google search. These rates get approximated experimentally by various companies. The specific values for the first 15 positions we use for our experiments are as follows.\footnote{These values are taken from \url{www.wikiweb.com/google-ctr/}, accessed January 20th, 2021.}
\begin{center} \small
\begin{tabular}{ lccccccccccccccc } 
 \toprule
  position\hspace{-0.8em} & 1 & 2 & 3 & 4 & 5 & 6 & 7 & 8 & 9 & 10 & 11 & 12 & 13 & 14 & 15\\  \midrule
 CTR & 32.5 & 17.6 & 11.4 & 8.1 & 6.1 & 4.4 & 3.5 & 3.1  & 2.6 & 2.4 & 1.0 & 0.8 & 0.7 & 0.6 & 0.4 \\
\bottomrule
\end{tabular}
\end{center}

\end{document}